\documentclass[12pt]{amsart}

\usepackage{cite}
\usepackage{xypic}
\usepackage{graphicx}
\usepackage{amsfonts}
\usepackage{mathrsfs}
\usepackage{amssymb}
\usepackage{amsxtra}
\usepackage{lscape}
\usepackage{pdflscape}
\usepackage{rotating}
\def\G{\Gamma}
\def\del{\partial}
\def\SS{\mathbb S}
\def\Z{\mathbb Z}
\def\T{\mathbb T}
\def\R{\mathbb R}
\def\C{\mathbb C}
\def\TTheta{{\mathbb T}_{\Theta}}
\def\TTn{\TTheta^n}
\def\Tn{\T^n_0}
\def\B{\mathscr B}
\def\BG{\bar \Gamma}
\def\Gsimp{\bar\Gamma_{simp}}
\def\Sk{\mathbb{S}_k}

\def\nn{\nonumber}

\def\A{\mathscr{A}}
\def\H{\mathscr{H}}

\def\chr{\Xi}

\newtheorem{thm}{Theorem}[section]

\newtheorem{cor}[thm]{Corollary}

\theoremstyle{definition}

\newtheorem{rmk}[thm]{Remark}

\newtheorem{ex}[thm]{Example}

\newcommand{\cev}[1]{\reflectbox{\ensuremath{\vec{\reflectbox{\ensuremath{#1}}}}}}

\accentedsymbol{\dbarG}{\Bar{\Bar{\Gamma}}}

\begin{document}


\title[Singularities, swallowtails and Dirac points]
{Singularities, swallowtails and Dirac points.
An analysis for families
of Hamiltonians and applications
to wire networks, especially the Gyroid
 }

\author
[Ralph M.\ Kaufmann]{Ralph M.\ Kaufmann}
\email{rkaufman@math.purdue.edu}

\address{Department of Mathematics, Purdue University,
 West Lafayette, IN 47907
 }

\author
[Sergei Khlebnikov]{Sergei Khlebnikov}
\email{skhleb@physics.purdue.edu}

\address{Department of  Physics, Purdue University,
 West Lafayette, IN 47907}

\author
[Birgit Kaufmann]{Birgit Wehefritz--Kaufmann}
\email{ebkaufma@math.purdue.edu}

\address{Department of Mathematics and Department of Physics, Purdue University,
 West Lafayette, IN 47907
 }

\begin{abstract}
Motivated by the Double Gyroid nanowire network we develop methods to detect Dirac points
and classify level crossings, aka.\ singularities in the spectrum
of a family of Hamiltonians. 

The approach we use is singularity theory. 
Using this language, we obtain a characterization
of Dirac points and also show that the branching behavior
 of the level crossings is given by an unfolding of $A_n$ type singularities.
 Which type of singularity occurs can be read off
  a characteristic region inside the miniversal unfolding of an $A_k$ singularity.

We then apply these methods in the setting of families of graph Hamiltonians, 
such as those for wire networks.
In the particular case of the Double Gyroid we analytically classify its singularities
and show that it
has Dirac points. This indicates that nanowire systems 
of this type should have very special physical properties.
\end{abstract}

\maketitle


\section*{Introduction}

In many physical situations one is led to a family of finite dimensional 
Hamiltonians defined over some parameter space (base) $B$ and would like to consider 
and classify the level crossings that appear when one changes the parameters. 
We came upon this general context when studying the $C^*$--geometry
of wire networks in general and the Double Gyroid network in particular \cite{kkwk,kkwk2}.
These networks are spatially periodic, and the base $B$ is a torus spanned 
by the values of the quasimomenta. The dependence of energy eigenvalues on quasimomenta 
determines the band structure; the simplest type of band crossing, 
a conical intersection,
is often referred to as a Dirac point. Band crossings are interesting because they may
be responsible 
for new physical phenomena (as well known, for instance, in the case of Dirac points in 
graphene). Of particular interest are Dirac points in triply periodic materials, such as
the Gyroid network: they can be viewed as magnetic monopoles in the 3-dimensional parameter 
space  \cite{Berry}  and as such are expected to be
stable under small deformations of the Hamiltonian.
In order to study various types of level crossings, we first widen the context to that 
of a family of Hamiltonians over an arbitrary base, 
and then apply the general results to our initial problem, in which the base is actually
the compact $n$--torus.

To be more specific, we 
consider a differentiable map from an $n$--dimensional manifold
$B$ to the set of $k\times k$ Hermitian
matrices for a fixed $k$. The case of 
interest will be $B=T^n$, the $n$-dimensional torus, and it can be thought of as the space of momenta.
Our results about the multiplicities in the spectrum are obtained using singularity 
theory \cite{arnoldbook}.
They are twofold. First we give an analytic way of finding all Dirac points,
by considering the energy levels as the zero set of a smooth function $P$ on $B\times \R$.
 Using
the Morse Lemma (Theorem \ref{morselemma}) we explain
that a Dirac point in this context and language is an isolated $A_1$ singularity with the 
signature $(-\dots-+)$ (or $(+\dots+-)$ depending on the sign of $P$) 
for the function $P$. 
This effectively uses an ambient space to embed
the conical singularity.

Then by considering the energy levels as a singular fibration over the base space of
momenta, we classify the possible singularities in the fibers. The fibration in question is 
the first projection $B\times \R\to \R$.
For this we again use singularity theory, more precisely that of 
the singularity $A_{k-1}$ and its miniversal unfolding. In particular,
we define a characteristic map of the base $B$ of the family to the base $\Lambda=\C^{k-1}$
of the miniversal unfolding.
From its image, which we call {\em the characteristic region}, one can read off many details, 
such as at what points degeneracies occur, and what their nature is.
Degeneracies occur precisely over the points of intersection of the characteristic region
with the discriminant locus.

This point of view allows us to classify the possible degeneracies as those appearing in the 
discriminant locus or swallowtail of the $A_{k-1}$ singularity.
These are known by a theorem of Grothendieck \cite{grothendieck} on 
the singularities appearing in the fibers of the miniversal unfolding
 associated to the singularities corresponding to  Dynkin diagrams. Namely,  these are precisely
those  obtained by deleting vertices (and all incident edges) of that 
of $A_{k-1}$ and hence are
of the form $(A_{n_1},\dots, A_{n_r})$ for suitable $n_i$. This corresponds to simultaneous
crossing of $n_1+1,\dots, n_r+1$ levels.
How these levels cross or equivalently how the singularity unfolds is encoded
in the characteristic map and can qualitatively be read off from the characteristic region.

One way to view this result is as a more precise and more general
version of the  von Neumann--Wigner 
theorem \cite{vNW}. Indeed for the full family of traceless $2\times 2$ Hamiltonians in their standard
parameterization over $\R^3$, we reproduce that the locus of degeneracy is of codimension $3$.
More precisely, there is only one point $0\in \R^3$ in the preimage of the characteristic map 
restricted to the discriminant.

If the family is more complicated however, our methods tell us where degeneracies
can occur and how the levels cross. In this case the codimension $3$ is not universally true 
any more.
Among the graphs we study, we exhibit families, where the codimension of the degenerate locus
is $2$, $3$ or $1$. The precise dimension count comes from the 
intersection of the discriminant with the characteristic region and the dimension 
of the fibers of the characteristic map. For isolated singularities such as Dirac points
one needs that
the particular fiber of the characteristic map over the corresponding point
in the discriminant is 0--dimensional. This is true 
in the von Neumann--Wigner case and this fact corresponds to the ``extra equations'' as we
explain.

Our main aim of application is 
the commutative and
 non--commutative $C^*$--geometry of wire networks \cite{kkwk,kkwk2} in their description
as graph Hamiltonians.
The input data for the general theory are a graph $\Gamma$, which is embedded in $\R^n$, 
the crystal graph, together with a maximal symmetry group $L$ isomorphic to $\Z^n$ and a constant magnetic field 2--form given by a skew symmetric matrix 
$\Theta$.
A fundamental role in the whole theory is played by the abstract quotient 
graph $\bar \Gamma:=\Gamma/L$.
This graph, together with the induced data of the magnetic field and the 
embedding, was used to define the Harper Hamiltonian and the relevant $C^*$ 
algebra $\B$. This algebra which we called the Bellissard--Harper algebra 
is the {\em minimal} $C^*$--algebra generated by the magnetic translations
corresponding to $L$ and the Harper Hamiltonian $H$.
In \cite{kkwk2} we showed that $\B$ embeds into  $M_k(\TTheta^n)$, the $k\times k$ matrices
over the non--commutative torus $\TTheta^n$.
Here $\Theta$ contains the information of the $B$--field for the lattice $\Gamma$ and
 $k$ is the number of vertices of $\bar \Gamma$ which is the number of sites in a primitive cell.

One intriguing aspect is that this
description is very useful even in the commutative case, that is in the 
absence of a magnetic field. The $C^*$ approach yields a family of finite dimensional
Hamiltonians parameterized over a base torus $T^n=S^1\times \dots \times S^1$.
Namely, in the commutative case $\Theta=0$ and 
 $\T^n_0$ is the $C^*$ algebra of complex valued continuous 
functions of  the $n$--torus $T^n=S^1\times \dots \times S^1$. In standard
 notation $\T^n_0=C(T^n)$.
Likewise, using the Gel'fand--Naimark theorem the Bellissard algebra $\B$ is also the
$C^*$ algebra of a certain compact Hausdorff space $X$; $\B=C(X)$, which as 
we showed in \cite{kkwk}
is a branched cover over $T^n$. Physically the base $T^n$ parameterizes the momenta and
the space $X$ is given by the energies of $H$ as these momenta vary. In this sense 
they give the energy bands of a one-electron system.

The central questions that arise are the following: At what points do we 
have degenerate eigenvalues in the spectrum and which of these points
are Dirac points? And, can these be read off from the  
graph $\bar \Gamma$ and its decorations, such as a spanning tree 
or weights with values in some $\TTheta^n$?

Given a specific graph with $C^*$--algebra valued weights on the edges, 
this is done by analyzing the function $P$ and the characteristic map from the base torus $T^n$
 to the miniversal unfolding of the $A_{k-1}$ singularity.
We apply  these considerations to the cases of the Double Gyroid wire network which was our 
 initial interest
and, to illustrate the concepts and the possible behaviors, we consider several other 
examples along the way including the wire networks obtained from the double versions of the P and D 
surfaces as well as the honeycomb lattice.

Our main result here is an analytic proof that the spectrum of the Gyroid has four singular fibers,
two of which are $A_2$ singularities and two of which are of the type $(A_1,A_1)$. We furthermore
show that the latter two are Dirac points. 
This is the first analytic proof of this fact.
The relevant family of Hamiltonians also arises in a different context \cite{Avron}. There
the authors found numerically that the singular points lie on the diagonal 
of $B=T^3$ (viewed as a cube with opposite faces identified) and obtained
the spectrum on the diagonal, see also  \S\ref{gyroidsec} and \cite{sym}. We note that although this shows that
on the {\em one--dimensional sub--family} of Hamiltonians given by the diagonal, there are
two triple degeneracies and two two--fold double degeneracies,  using
this information alone one
cannot conclude how the singular structure extends onto the full 3--dimensional torus. Our method gives this extension. 
In particular, it allows us to show analytically that there are no degeneracies anywhere else in $T^3$.

Applying our program to the honeycomb lattice, 
we rediscover the well--known Dirac points of 
graphene \cite{Wallace}. Thus it can be hoped  that the Dirac points 
of the Gyroid will give rise
to important new material properties.

Another natural question in the  case of
graph Hamiltonians, 
which  is  addressed in a separate 
paper \cite{sym}, is if there are symmetries
that can be derived from the graph setup which explain the degeneracies.
The short answer is that (a) the symmetry group must be extended beyond the
permutation symmetry of the graph to include phase transformations on the vertices
(``re-gaugings''), and 
(b) in all the wire network cases corresponding to the 2 and 3 dimensional
self--dual graphs:
 P, D, G and honeycomb these symmetries force all the singularities.

\section{Singularities in spectra of families of Hamiltonians}
\subsection{Singularities}
We briefly recall pertinent definition of  singularity theory \cite{arnoldbook}.
In this theory one considers germs of smooth functions $f:\C^n \to \C$ with critical point at ${\bf 0}$
with critical value $0$ up to equivalence induced by germs of diffeomorphisms.
That is the germ $(f,{\bf 0})$ is equivalent to a germ $(f',{\bf 0})$
if  there exists a germ $(g,0)$ of a diffeomorphism $g:\C^n\to \C^n$ with $g({\bf 0})={\bf 0}$ such that
$f=f'\circ g$.
A {\em singularity} is  an equivalence class of such germs.
The germ $(f',0)$ is then also called the pull--back under $g$.

Analogous definitions hold in the case of complex singularities, that is for germs
functions $f:\C^n\to \C$
with diffeomorphisms replaced by biholomorphic maps.

A deformation of a germ $f:\C^n\to \C$ with base $\Lambda=\C^k$ is the germ at zero of a smooth map
$F:(\C^n\times \C^k\to \C,0)$ which satisfies $F(x,0)=f(x)$.
A deformation $F'$ is equivalent to $F$ if there is a smooth germ of a
diffeomorphism $g:\C^n\times\C^k\to \C^n\times \C^k$ at zero with $g(x,0)=x$.

Given a deformation $F$ with base $\Lambda$ a smooth germ $\theta:(\C^r,0)\to (\Lambda,0)$ {\em induces} a deformation $F'$ via pull--back. $F'(x,\lambda'):=F(x,\theta(\lambda))$.

A deformation $F$ of a germ $f$ is {\em versal} if every deformation $F'$ of $f$ is equivalent to a deformation induced from $F$. It is called {\em miniversal} if $\Lambda$ is of minimal dimension.

Again one can replace $\C$ with $\R$.
Also, one can pull--back not only germs, but actual pointed families by
the same procedure. Here the base spaces $\C^k$  
are then simply replaced by  smooth manifolds $B$. This yields the
same local theory.

\subsection{The spectrum as a zero locus}
\label{zerolocsec}
The basic starting point of our analysis in this section is that
given a smooth family of Hamiltonians over a base $B$ the spectrum as
functions on $B$ 
can be alternatively given as the zero locus of a single function $P$ on 
$B\times \R$.
The map that associates to a point $b$ in the base $B$
 the Hamiltonian $H(b)$    is a smooth
map $H:B\to Herm(k)$,
where $Herm(k)$ are the Hermitian $k\times k$ matrices. Now since
the matrices depend differentiably on the parameters, varying these eigenvalues
gives rise to a cover $\pi:X\to B$. Here a priori the cover is just that of
sets, but one can quickly show that this is a cover of topological spaces, e.g.\ using $C^*$ geometry, see \S\ref{coversection}.

The inverse image of a point $B$ under $\pi$
is the set of eigenvalues of $H(b)$. 
As these are the zeros of the characteristic polynomial
of $H(b)$, we obtain another description of $X$ as a subset in 
an ambient differentiable manifold as follows.

Consider the trivial cover $B \times \C \to B$ and its real part
$B\times \R\to B$. On the space $B\times \C$, we
consider the function $P:B\times \C\to \C$ given by $P(b,z):=det(zId-H(b))$.

The zero locus of $P$ is exactly $X$. 
We chose the normalization, so that $P$ starts
with $+z^k$. Fiberwise for the zero locus,
we just get the eigenvalues of $H(b)$,
and since $H(b)$ is Hermitian, we know that we have real  eigenvalues and
hence $X$ is also the zero locus of $P(z,b)$  contained in $B\times \R$.

The description above gives $X$ as a singular manifold. Around points of $X$
at which  $P$ does not have a critical value $0$,
$P^{-1}(0)$ is a smooth manifold. 
The points at which $P$ is critical with critical value $0$ are singular.

Physically $P^{-1}(0)$ are just the energy levels.
A level crossing can occur only at critical points of $P$ with critical value $0$.
Namely, if we fix $(b,z)\in B\times \R$
 and $0$ is a non--critical value in a small neighborhood $U$ of $(b,z)$ 
then $P|_U^{-1}(0)$ is a smooth manifold.
More precisely, if $\frac{\del P}{\del z}\neq 0$ the implicit function 
theorem states that
for $(b_0,z_0)\in U$ with $P(b_0,z_0)=0$ there exists a 
function $z=E(b)$ such that $P(b,E(b))=0$
and the graph $z=E(b)$ is the component of the 
smooth manifold $P|_U^{-1}(0)$ containing $(b_0,z_0)$. In other words:
$E(t)$ is the dispersion relation.

Our main results describe the singular locus of the singular manifold $X$. 
There are two approaches
we will take. First we can look at the singularities 
of $X$ locally where we regard $X$ as embedded in $B\times \R$. 
We use this to find Dirac points, see \S \ref{diracsection}. 
Secondly, we can look at the cover $\pi:B\times \R\to B$
restricted to $X$, that is $\pi:X\to B$. The upshot is that locally around
a singular point $x$, $X$ is the  deformation of the singularity 
of the fiber over $\pi(x)$.
By Grothendieck, locally in a fiber 
the only  singularities that can appear are of type $A_n$ with $n\leq k-1$.
This point of view lets us classify these singularities and their deformations
by means of the miniversal unfolding of the $A_{k-1}$ singularity, 
see \S\ref{swallowsection}.

\subsection{Dirac points as Morse or $A_1$ singularities}
\label{diracsection}
From the point of view of material properties, one of the most interesting singularities that can occur are Dirac points.
In the general terminology, a Dirac point is a conical singularity in the spectrum when two levels cross such that the dispersion
 relation is linear.

 In our situation, this can be formalized in order to yield analytical tools to find and classify these points without solving the Eigenvalue equations. Instead of actually finding an expansion,
  we will use a smooth ambient space to characterize Dirac points using singularity theory.

 The standard cone is of the form $z^2=\sum_{i=1}^k t^2_i$. Here as above we fix the sign
of the coordinate $z$ to be positive.
 We can rewrite this as $F({\bf t},z)=0$ where $F({\bf t},z)=z^2-\sum_{i=1}^k t_i^2$.
The characteristic features of the function are that $f$ has a critical point at
${\bf 0}=(0,\dots,0)$ with critical value $0$.
Moreover the Hessian, i.e.\ the matrix of its second derivatives, is a quadratic form with signature $(-\dots-+)$. In particular
its determinant $hess=det(Hess)\neq 0$. In general if $f$ has a critical point and $hess\neq 0$
the critical point is called a Morse critical point.

The most pertinent theorem about Morse critical points is the Morse Lemma.
\begin{thm}
\label{morselemma}
\cite{morsebook}
In a neighborhood of a nondegenerate critical point $p$ of a smooth function $F:M\to \R$ from an $n$--dimensional manifold $M$
there are coordinates
$x_i$ centered at $p$, such that in these coordinates 
$$F=-x_1^2-x_2^2-\cdots -x_\lambda^2+x_{\lambda+1}^2+\cdots +x_n^2+f(p)$$ 
where, $\lambda$ is the index of the critical point. 
\end{thm}

Now we see that a Dirac point as a germ is equivalent to the  germ of $(f,0)$
above, whose index $\lambda$ is $n-1$ or $-1$ if one switches the sign of $F$,
that is if one regards $-F=0$. Let us for the moment assume that the sign of $F$ 
is chosen such
that the signature is $(-\dots-+)$ in the order of coordinates above.

Such a germ is  the pull--back under some diffeomorphism on the ambient space
and the cone itself is the zero set of the function $F$.
Notice that the characteristic properties of being a 
Morse critical point are invariant
under the diffeomorphism. 
The signature has a geometric meaning. It says that the cone 
 opens up on the $x_n$--axis.
The dispersion 
relation in the new coordinates is just the pull--back and
hence also linear.

In our particular case, that of the fibration  $B\times \R\to B$ the function $F=P$ and the role 
of $x_{n}$ should basically be that
of the coordinate $z$ on the  fiber $B\times \R\to B$
and $(x_1,\dots ,x_{n-1})$ should correspond
to the  variables on the n--dimensional base, so that we indeed
get a physically sensible dispersion relation $z=E(b)$.
Since we are working with germs,
we do not distinguish between a local neighborhood in $B$ and the image
of its local charts in $\R^n$.

Now the cone has the right orientation
 as long as coefficient $\frac{\del^2P}{\del z^2}>0$, as this
states that the $z$--direction lies in the positive part of the cone.
  The dispersion relation
can depend on the direction that is the cone might be deformed and tilted. The
tilting of the cone is given by the partial derivatives
 $\frac{\del^2P}{\del z\del b_i}$. The cone is
untilted if  $\frac{\del^2f}{\del z\del b_i}=0$ for the base coordinates $b_i$.
Rephrased, we need that $T(B\times \R)=TB\oplus T\R$ is a decomposition into
a negative definite and a positive definite subspace.

Recall that a stabilization of a singular germ $f({\bf z})$ is a function $f({\bf z},{\bf w})\pm w_1^2 \pm \dots \pm w_m^2$. Thus the Morse singularities 
are then stably equivalent to the $A_1$ singularity given by the 
germ $f(z)=z^2$.

\subsubsection{Singularity Characterization of Dirac Points}
Therefore, we get the Dirac points in the 
spectrum are precisely the critical points 
with critical value $0$
of $P$ which are stabilizations of an $A_1$ singularity 
with signature $(-\dots-+)$ or $(+\dots+-)$ in the coordinates
$(b_1,\dots,b_n,z)$ such that  $T(B\times \R)=TB\oplus T\R$ is a decomposition into
a negative definite and a positive definite subspace.
We can now take this as their {\em definition}.

Practically this means that we have to simultaneously 
solve the equations $P=0,\nabla P=0$
and then check $hess\neq 0$ and moreover check that the signature is correct, 
by computing the principal  minors and check that $\frac{\del^2P}{\del z^2}=0$.

\label{swallowsection}
\subsection{The spectrum as a pull--back from the miniversal unfolding of the $A_{k-1}$ singularity}

 \subsubsection{The miniversal unfolding of $A_{k-1}$ and the swallowtail}
The $A_{k-1}$ singularity is the singularity defined by the function $f(z)=z^k$, 
which has a critical point of order $k-1$ at $0$. 
Its miniversal deformation \cite{arnoldbook}  is
\begin{equation}
\label{Aneq}
F(a,z)=z^k+a_{k-2}z^{k-2}+\dots +a_0
\end{equation}
 According to the general theory, the dimension of a miniversal deformation
  coincides with the
 dimension of the Milnor ring $\C[z]/(f')$ (also called the Milnor number) 
and the terms which are added to $f$ are in 1-1 correspondence 
with the vector space basis $(1,z,z^2,\dots,z^{k-2})$ of this ring.

 The geometry of the situation is very similar to the cover $\pi$ considered introduced in
\S\ref{zerolocsec}. In particular, the function $F(a,z)$ is a function on 
$\C^{k-1}\times \C$. 
Let $Y:=\{(a,z):F(a,z)=0\}\subset \C^{k-1}\times \C$ and
 considering the trivial bundle $\C^{k-1}\times \C\to \C^{k-1}$. 
Again we get a branched 
cover $\pi:Y\to \C^{k-1}$.
  The inverse image under $\pi$ is the set of  roots of the polynomial. 
Generically there are $k$ of these
  roots. However, over a subset $D\subset \C^{k-1}$ of the base space 
the number of inverse images drops as
 there are multiple roots. This set is known as the discriminant locus, 
the swallowtail or the level bifurcation set  and has been extensively 
studied (see \cite{arnoldbook,GKZ}). 
It is the zero set of the discriminant of the polynomial $F(z):=F(z,a)$ which is
  considered as a polynomial with arbitrary coefficients. 
The discriminant is a simple polynomial in the $a_i$
  and its zero set has codimension $1$ \cite{GKZ}.
   \begin{figure}
\includegraphics[width=.3\textwidth]{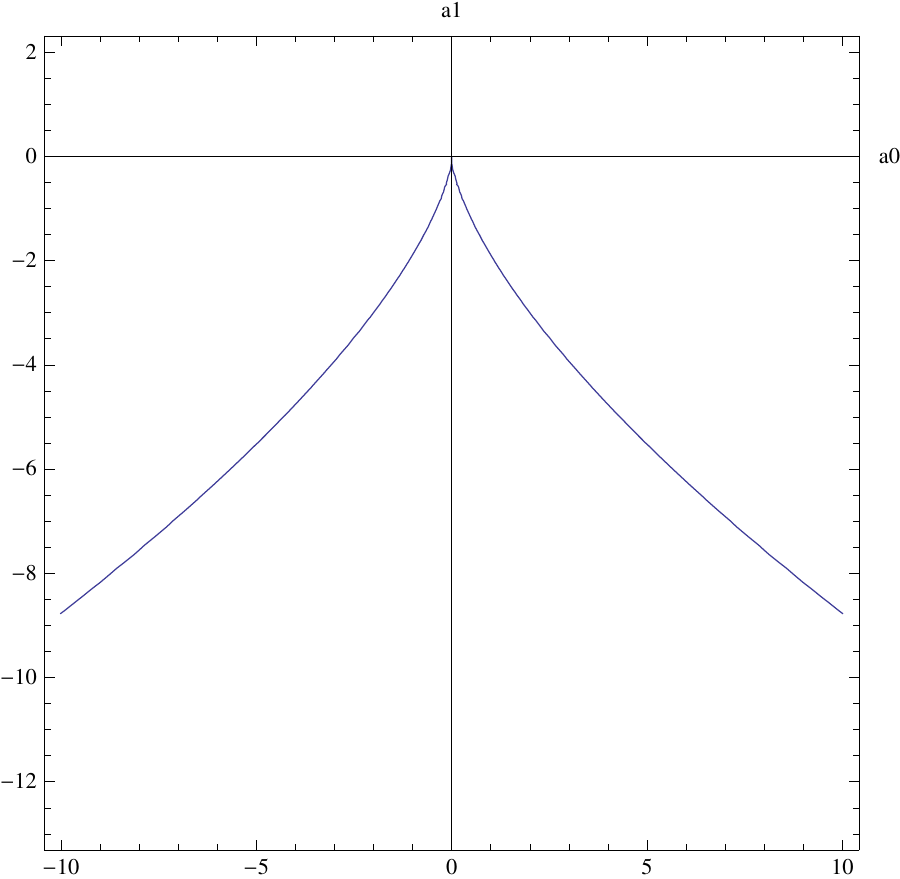}
\hspace{1cm}
\includegraphics[width=.5\textwidth]{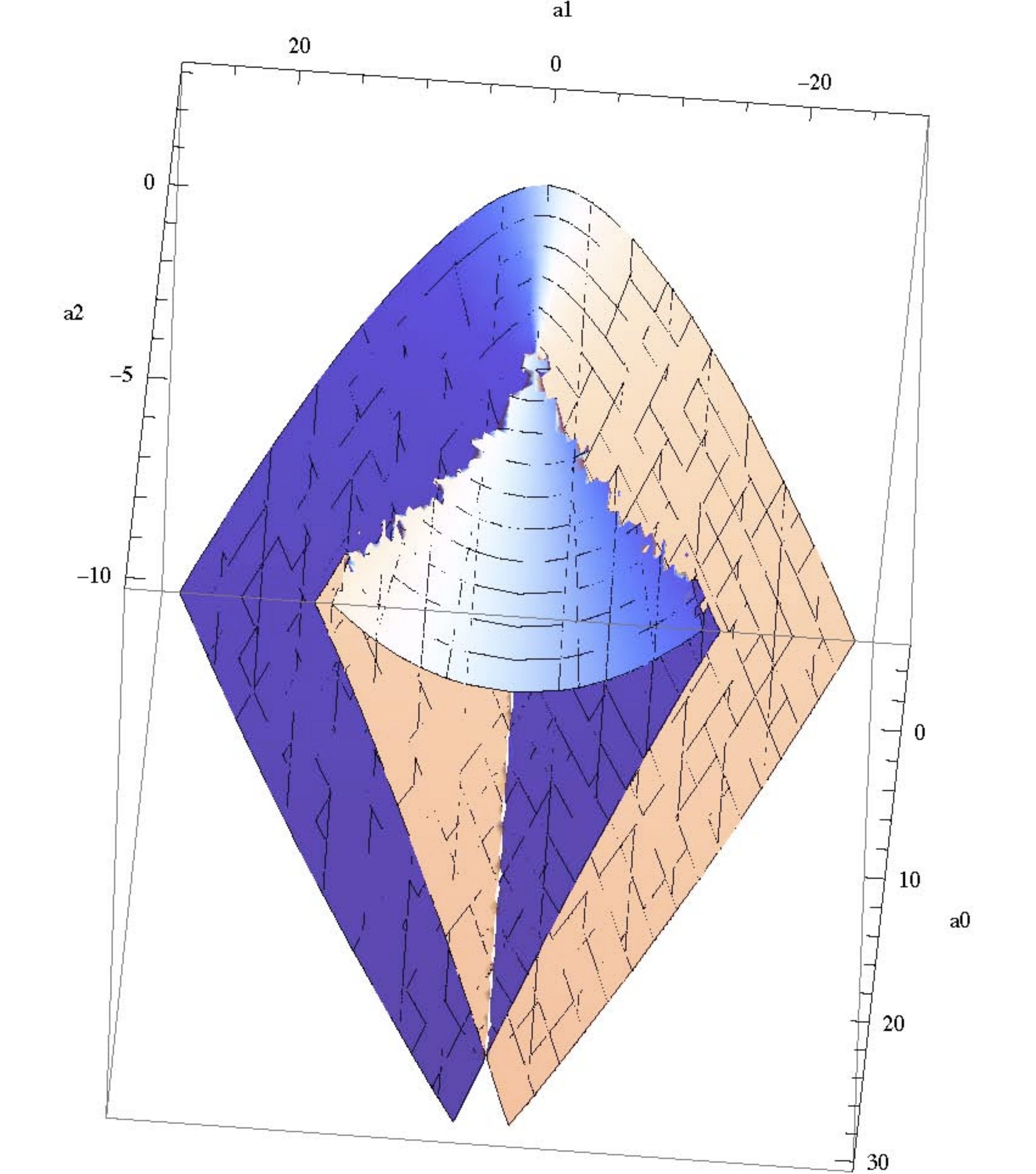}
\caption {Zero locus in the $A_2$ and $A_3$ singularities}
 \label{sing}
\end{figure}

 Pictures for this locus in the $A_2$ and $A_3$ are shown in Figure \ref{sing}. 
The surface $D$ in the $A_3$ singularity is known as the swallowtail. 
In general $D$ goes by the names of discriminant, level-bifurcation set
 or also again as the swallowtail.

\subsubsection{The spectrum as a pull--back}
\label{mainthmsec}
 Consider the trivial bundle $B\times \C\to B$ as before
and let $P(b,z)$ be defined as before. We expand

\begin{equation}
\label{peq}
P(b,z)=z^k+a_{k-1}(b)z^{k-1}+a_{k-2}(b)z^{k-2}+\dots +a_0(b)
\end{equation}
where now the coefficients $a_i:B\to \R$ are real since the matrix $H(b)$
is Hermitian. This has great similarity to Eq.\ (\ref{Aneq}), except for
the second leading term not vanishing.

 However a simple invertible
smooth transformation $s:z\mapsto z-a_{k-1}/k$ yields a polynomial of this type. 
In our setup, the transformation $s$ gives
a diffeomorphism $g:B\times \R \to B\times \R$
 given by $(b,z)\mapsto (b,z -a_{k-1}(b)/k)$. This  gives an 
equivalence between $P$ and $\hat P:=P\circ g$.
Now $\hat P$ expands as

\begin{equation}
\label{phateq}
\hat P(b,z)=z^k+\hat a_{k-2}(b)z^{k-2}+\dots +\hat a_0(b)
\end{equation}

Let $\chr=(a_0,\dots,a_{k-2}):B\to \R^{k-1}$,
 then the miniversal
unfolding $F$ pulls back via $\chr$
to  the  deformation $\chr^*(F):B\times \C\to \C$ given 
by $\chr^*(F)(b,z)=\hat P(\hat a(b),z)$. In other words $\hat P$ is the 
deformation induced from the miniversal deformation $F$ via $\chr$.

We define {\em the characteristic map}  to be the coefficient map
$\chr:B\to \C^{k-1}$ and the {\em  characteristic region $R$} 
 to be the image of $\chr$.

Let $\chr^*(Y)\subset B\times \C$  be the zero locus of $\chr^*(F)$, 
then we get  a pullback of  
the map $\pi$: $\chr^*(\pi):\chr^*(Y)\to B$ by restricting the projection.

Notice that if the Hamiltonians $H(b)$ are traceless then $a_{k-1}\equiv 0$ and
$P=\hat P$.
Summing up we have the following:
 \begin{thm}
 \label{mainthm}
 The branched
cover $X\to B$ is equivalent via $g$ to the pull back of
 the miniversal unfolding of the $A_{k-1}$
 singularity along the characteristic map $\chr$.

Moreover if the family of Hamiltonians is traceless, the cover
is the pull--back on the nose.
\end{thm}

\begin{cor}
The possible degeneracies are pull-backs of
those which appear in the miniversal unfolding of $A_{k-1}$.
Moreover, these singularities are present over the fibers over the real part of the discriminant.
\end{cor}

Since $A_{k-1}$ is a simple singularity the following theorem of Grothendieck applies.
\begin{thm}\cite{grothendieck}
The types of  singularities which appear in the swallowtail of a simple singularity are exactly those
corresponding to the Dynkin diagrams obtained by deleting vertices and all the edges incident
to these vertices in the Dynkin diagram of the original singularity.
\end{thm}
If the resulting diagram is disconnected this means that there are several critical points with critical values $0$ in the fiber.
In fact there is a stratification of the swallowtail $\Sigma$ into strata $\Sigma(X_1,\dots, X_n)$
according to the types $X_i$. The above theorem then states which strata are non--empty.

If we take the $A_k$ Dynkin diagram and delete $l$ points, there are at most $l+1$ connected components
all of type $A_r$ for some $r$ and the sum of their Milnor numbers $r$ is at most $k-l$. Deleting
points at the edges or next to each other, we can make the individual Minor numbers smaller.

\begin{cor}
The only possible types of singularities for nonloop graphs in the spectrum are $(A_{r_1}, \dots, A_{r_s})$ with $\sum r_i\leq k-s$.
\end{cor}

\begin{ex} In the unfolding of the $A_3$ singularity, we have an $A_3$ singularity at the origin,
we have $A_1$ singularities along the smooth part of the swallowtail corresponding
to deleting two points of $A_3$. Along the cusps of the swallowtail there are $A_2$ singularities
corresponding to the triple degeneracy
 of the roots corresponding to the right two vertices and the left two vertices and over the double points there are $(A_1, A_1)$ singularities  corresponding to deleting the middle point of $A_3$.
\end{ex}

\begin{cor}
Near any singular point $x\in X$ of $P$ there is a 
neighborhood of $x$ which is a deformation of
an $A_r$ singularity for some $r\leq k-1$.
\end{cor}

\begin{proof}
If $x$ is a singular point, then $y=\chr(\pi(x))$ lies on the discriminant. 
Picking a neighborhood
of $y$ and pulling it back to $B$ and to $X$, we obtain the desired deformation
by restricting to the component that $x$ lies in.
\end{proof}

\subsection{Characteristic region}
Since the families we consider are given by Hermitian matrices,
they have positive eigenvalues. Hence, if $disc$ is the discriminant
function of $F$, then $disc\circ\chr\geq 0$ as this function is
is quadratic in the differences of the eigenvalues. Thus, we
get that the characteristic region is contained in the locus of $\Lambda$
over which the discriminant is non--negative.

 Notice that
if $B$ is compact connected, then the image under $\chr$ of $B$ 
is compact connected. This will be the case for graph Hamiltonians
where $B=T^n$, hence
 it will then also lie in the closure of a component of
  the real part of $\R^{n-1}\setminus D\cap R^{n-1}$ 
over which the discriminant is positive.

Since the discriminant is a simple polynomial,
the sign changes when crossing the discriminant locus.
This entails that the intersection of $D$ with the characteristic region
is only along its boundary $Bd(R)$. 
Thus if $p$ is an interior point of $R$
then all the fibers of $X$ over the inverse images 
of $p$ under $\chr$ are non--singular
and moreover there is a whole non--singular neighborhood of each fiber.

The fibers of $X$ that are singular all lie over points $b$ whose image 
$\chr(b)$
sits in $D\cap R\subset Bd(R)$.
Thus the characteristic region gives a useful aid in studying the singularities that occur.
If $R\cap D=\emptyset$ then there are no singularities. 
In low dimensions this is also
a great visualization tool.

Notice that if $B$ is connected, isolated intersection points only occur for
constant maps, which means that there are no degrees of freedom. 
This cannot happen in the crystal/wire case.

If there are non--isolated intersection points, 
then their inverse images are singularities. 
Namely, by the above, there are nearby fibers, 
where the number of pre--images is $k$.

Over a fiber in a component with an  $A_r$ singularity,
we hence know that there are transversal directions, 
where one Eigenvalue splits into $r+1$ different pre--images. Or
going into the singularity $r+1$ energy level coalesce.
Thus over a point in the stratum $(A_{n_1},\dots, A_{n_l})$ there  are
$k$ crossings of $n_1+1,\dots,n_l+1$ levels, respectively.

\subsection{Consistencies and necessary conditions}
There are several consistencies which one might exploit 
for analytic or numerical solutions.
Any pre--image $b$ of a point $p$ in the boundary 
of $R$ has to satisfy that  $J_{\chr(b)}$,
the Jacobian of $\chr$ at $b$, does not have maximal rank. 
This is of course not a sufficient condition.
We know that to get a singular point we need that $J_{\chr}(b)=0$. 
This again fits well with the
fact that the discriminant restricted to $R$ is $\geq 0$. 
So that $disc\circ \chr$ has a zero Jacobian.

In order  for an isolated singularity at $x\in X$, such as a Dirac point, 
to occur, a necessary condition is that the fiber
$\chr^{-1}(\chr(\pi(x)))$ is discrete. 
This takes care of the vertical direction, 
but of course there should also be no curve through $x$ 
transversal to the fibers mapping to the discriminant.
All these types of behaviors can be found in the examples we give.

Another nice consistency check is given by the discriminant of $P$ considered
as a polynomial in $z$. 
Since any singular point $x$ in the spectrum is a critical 
point with critical value zero, we must have $P(x)=P_z(x)=0$,
which means that discriminant $disc(P)(x)=0$. 
Denoting the discriminant of the $A_k$ singularity  by $disc$ as well,
we have  $disc(P)(x)=disc(\chr(\pi(x)))=0$.

\subsection{Standard von Neumann--Wigner Example}
We consider the family of Hamiltonians $H(a,b,c)=a\sigma_x+b\sigma_y+c\sigma_z$,
where $\sigma_x,\sigma_y,\sigma_z$ are the Pauli matrices.
This gives us a family with base $\R^3$. It is the full family of traceless Hermitian $2\times 2$ matrices.

The usual interpretation of von Neumann--Wigner is that for a single level crossing one can reduce to this
$2\times 2$ family. However, this is basically only true in a ``generic'' or abstract setting and not for any arbitrary particular family. 
The original article \cite{vNW} does not claim this, but rather computed the co--dimension of the space of Hermitian matrices with degenerate 
eigenvalues in the whole space of Hermitian matrices. This is where the pure dimension count takes place. 

In our setting the calculation proceeds as follows.
The function $P$ on $\R^3\times \R$  is given by $P(a,b,c,z)=z^2-a^2-b^2-c^2$. The singularity is 
$A_1$ which has a one dimensional base $\Lambda=\R$. The characteristic map is  the map $\chr:(a,b,c)\mapsto -a^2-b^2-c^2$.
The discriminant is just the point $0\in \R$ and we see that we have a level crossing over $\chr^{-1}(0)=(0,0,0)$.
And indeed, the inverse image is  zero dimensional and the codimension of its locus is $3$.
All other fibers of $\chr$ are of codimension $1$ as one would expect from just a count of equations.
Here we neatly see how the na\"ive equation count fails over the special fiber.

The extra dimension drop can be explained using Picard--Lefschetz theory as a vanishing sphere.
Likewise the ``diabolical'' nature of these points ---that is the behavior of wave functions when moved around the conical singularity---
can be explained via the classical monodromy operator \cite{arnoldbook}.

The fact that the singularity is conical can readily be checked in our framework. Indeed $P$ has an isolated critical point at $(0,0,0,0)$ 
with value $0$ and signature $(---+)$ of the Hessian.

Finally, if one looks at the even larger family,  $H(a,b,c,d)=a\sigma_x+b\sigma_y+c\sigma_z+d\,Id$ of all Hermitian $2\times 2$ matrices,
then one gets a family over $\R^4$ with $P(a,b,c,d,z)=z^2-2dz+d^2-a^2-b^2-c^2$. So in this situation one has to use the
shift $s: z\to z+d$ upon which $\hat P(a,b,c,d,z)=z^2-a^2-b^2-c^2$. The characteristic map $\chr(a,b,c,d)=-a^2-b^2-c^2$,
and we see that singularities are over $\chr^{-1}(0)=(0,0,0,d)$ so that the fiber is now one dimensional, but
still of codimension $3$. The singular locus in the spectrum is given by a conical singularity crossed a line and hence
not isolated. Indeed the Jacobian of $P$ vanishes along the line $(0,0,0,d,d)$ with critical value $0$
and the Hessian of $P$ is degenerate at these points.

\section{Graph Hamiltonian and Wire network setup}

We now discuss the families that come about by considering Hamiltonians obtained
from finite graphs with (commutative) $C^*$ algebra weights on the edges. These
in turn arise from wire networks of real materials. Here the $C^*$--algebra in question
is the noncommutative n--torus $\TTheta^n$ where $\Theta$ is a skew symmetric matrix that
 encodes the commutation relations of the $n$ unitary generators $U_i$ of $\TTheta$.
Physically it corresponds to a constant magnetic field $B$.

The initial setup for graph Hamiltonians works with a general $C^*$--algebra $\mathscr A$,
the relevant Hilbert space being an $\ell^2$ space.
In case  $\mathscr A$ is commutative, by the Gel'fand--Naimark theorem, we get another description
of the algebra that $\mathscr A$ and the Hamiltonian generate
as the set of levels of a family of finite dimensional Hamiltonians parameterized over a base.
For concreteness, we will set  ${\mathscr A}=\TTheta$, but
also comment on how to treat the general case.

\subsection{Hamiltonian from a finite graph}
In order to set up the general theory for graph Hamiltonians, 
we fix a finite graph $\bar \Gamma$,
a rooted spanning tree $\tau$ of $\bar \Gamma$, an order $<$ of the vertices of $\G$ such that
 the root of $\tau$ is the first vertex, a skew symmetric matrix $\Theta$, and a morphism
$w:\{\text{Directed edges of } \G\}\to \TTheta^n$ which satisfies the following
\begin{enumerate}
\item $w(\vec{e})=w(\cev{e})^*$
if $\vec{e}$ and $\cev{e}$ are the two orientations of an edge $e$.
\item $w(\vec{e})w(\cev{e})=1$

\item $w(\vec{e})=1\in \TTheta^n$ if the underlying edge $e$ is in the spanning tree.
\end{enumerate}

Let $k$ be the number of vertices of $\Gamma$. We will enumerate
the vertices $v_0,\dots, v_{k-1}$ according to their order; $v_0$ being the root.
Given this data, the Hamiltonian $H=H(\Gamma,\tau,<,w)\in M_k(\TTn)$ is the $k\times k$ matrix whose
entries in $\TTheta^n$ are 

\begin{equation}
H_{ij}=\sum_{\text{directed edges $\vec{e}$ from $v_i$ to $v_j$}} w(\vec{e})
\end{equation}

\subsection{Family in the commutative case}
\label{coversection}
 If $\Theta=0$ we can consider these Hamiltonians
as a family of Hamiltonians over the  n--torus $T^n$ as follows.
First $\T^n_0$ is the $C^*$ algebra of continuous $\C$--valued functions on $T^n$ via
the Gel'fand--Naimark correspondence.
Considered as a space each point of $\T^n_0$ is given  by a character or $C^*$--algebra morphisms
$\chi:\Tn\to \C$. To each such a point, we associate the Hamiltonian $\hat \chi(H)\in M_k(\C)$
where $\hat \chi$ is the natural lift of $\chi$ to the matrix ring. Specifically,
\begin{equation}
(\hat \chi(H))_{ij}=\chi(H_{ij})
\end{equation}
The correspondence between points and characters is in the following way. 
Each point $t\in T^n$
gives rise to the evaluation map  $ev(t):C^*(T^n)\to \C$ which sends a function $f$ to its value $f(t)$ at $t$. Varying $f$ we obtain
a character. The Gel'fand--Naimark theorem asserts that this is 1--1.

Using this correspondence, we get a Hamiltonian $H(t)$ for each $t$. Physically if we think that $T^n$ parameterizes
momenta, we obtain $H(t)$ by just plugging in the given momenta.

Using the formalism we developed, we  get a topological cover $X\to T^n$ where the points over a base point $b$ are
the eigenvalues of $H(t)$ and furthermore a realization of this cover as a subspace
in the manifold $T^n \times \R$ and all of our analysis applies.

\subsubsection{$C^*$--geometry}
One can understand the topological cover $X\to T^n$ in $C^*$--geometry which yields
the basic connection  of our analysis in this section to the previous one. Consider $\B_0\subset M_k(\Tn)$, 
 the algebra generated by $H\in M_k(\Tn)$ and the diagonal embedding of $\Tn$ into $M_k(\Tn)$ as scalars.
This algebra is still commutative and
again by
 applying the Gel'fand--Naimark theorem, we obtain a   compact Hausdorff space $X$, such that $\B_0$ is $C^*(X)$.
The main point is that
the cover of the torus given by the $C^*$ analysis  from the inclusion $\Tn\to \B$  (see \cite{kkwk})
is exactly the cover $\pi:X\to T^n$ considered in the last section.

\begin{rmk}
One can readily generalize this situation to any commutative unital $C^*$ algebra $\A$. We then
get a Hamiltonian over the base space $B$ which satisfies $C^*(B)=\A$.
The role of the algebra
$\B_0$ is then played by the algebra in $M_k(\A)$ generated by $H$ and the diagonal embedding of $\A$.
\end{rmk}

\subsection{Further characterization of $P$}
Again considering $P$ as a polynomial in $z$, 
the coefficient functions $a_k$ in equation (\ref{peq}) can be given a graph theoretical interpretation.
For this it is convenient to introduce the graph $\Gsimp$ and the weight function $w^+$ associated
to $(\bar\Gamma,w)$. The vertices of $\Gsimp$ are just the vertices of $\Gamma$. The edges
of $\Gsimp$ are simply the equivalence classes of edges of $\bar\Gamma$, where two edges are equivalent if they run between the same vertices. This identification induces a weight function $w^{+}$ on $\Gsimp$,
where now $w^+(\vec{[e]})=\sum_{\vec{e'}\in [e]} w(\vec{e'})$. That is the sum over all edges connecting
the same two vertices as $e$. In this notation

\begin{equation}
H_{ij}=\begin{cases} 0 & \text{ if there is no edge between $v_i$ and $v_j$ in $\Gsimp$}\\
w^+(\vec{[e]})&  \text{ if there is a necessarily unique oriented edge $\vec{[e]}$}\\
 &\text{ from $v_i$ to $v_j$ in $\Gsimp$}\\
\end{cases}
\end{equation}

 Plugging this into the usual determinant formula $$det(A)=\sum_{\sigma\in \Sk} sign(\sigma)a_{1\sigma(1)}\cdots a_{k\sigma(k)}$$
 we can  give the summand corresponding to  $\sigma$ graph combinatorially. Decompose $\sigma$ into cycles $c_1,\dots, c_q$ of length $l_1,\dots, l_q$. Then each cycle corresponds to a unique cycle
of oriented edges in $\Gsimp$. Explicitly if $c_j=(j_1j_2\dots j_{l_j})$ then the cycle of $\Gsimp$ is
given by the unique directed edges from $v_{j_r}$ to $v_{j_{r+1}}$ and from $v_{j_{l_j}}$ to $v_{j_1}$ if
all these edges exist. In that case and if $l_j>1$ we set $w^+(c_j)=\prod w^+{\vec{[e]}}$
where the product runs over all the oriented edges in that cycle. Otherwise set $w^+(c_j)=0$.
If $l_j=1$ then $w^+(c_j)=-z+w^+\sum_{e}(\vec{e})+w^+(\cev{e})$ where $e$ are loop edges from
$v_{j_1}$ to itself if these exists and if there are no such edges,  set $w^+(c_j)=-z$ otherwise.
In this notation:
\begin{equation}
\label{patheq}
P(t,z)=(-1)^k\sum_{\sigma\in \SS_k}sign(\sigma) p_{\sigma}(t,z) \text{ with } p_{\sigma}(t,z) =\prod_{j=1}^q w^+(c_j)
\end{equation}

\subsubsection{Graphs with no  small loops }
Assume that $\BG$ has no small loops, that is edges which return to the same vertex.
 Then all the diagonal entries $H_{ii}=0$ and hence $a_{k-1}=0$.
Furthermore, if $i$ be the number of cycles of length one and we assume that these
are the first $i$ cycles, then
\begin{equation}
\label{pathnoloopeq}
p_{\sigma} =(-z)^i\prod_{j=i+1}^qw^+(c_j)(t)
\end{equation}
where the product is now over the cycles of length $>1$.
One can also again read off that $a_{k-1}=0$. Namely, if $k-1$ cycles have length $1$ then all $k$ cycles
have length one.

Using the formula (\ref{pathnoloopeq}) the coefficient $a_{k-2}$ becomes 
\begin{equation}
\label{aktwoeq}
a_{k-2}=-\prod_{e\in \Gsimp} w^+(\vec{e})w^+(\cev{e})
\end{equation}

 \subsubsection{Simply laced graphs with no small loops}

 Thus if furthermore $\bar\Gamma$ is simply laced, that is there is at most one edge between two vertices, then $a_{k-2}=|E(\BG))|$ is simply the number of edges.
Summing up in this case:
\begin{equation}
P(t,z)=z^k-|E_{\BG}|z^{k-2}+a_{k-3}(t)z^{k-3}+\dots +a_0(t)
\end{equation}

Applying the results of \S\ref{mainthmsec} in this situation yields the following.
 \begin{thm}
 If the graph $\BG$ has no small loops the branched cover $X\to T^n$ is the pull back of
 the miniversal unfolding of the $A_{k-1}$ singularity along the 
characteristic map $\chr:(a_0(t),\dots,a_{k-2}(t))$. 
Otherwise it is equivalent to the pull--back.

 The characteristic region is compact connected and lies in the 
closure of a component of
  the real part of $\R^{n-1}\setminus D\cap R^{n-1}$ over which the 
discriminant is positive.

 Moreover if $\BG$ has no small loops and is simply laced, then $R$ 
is contained in the hyperplane 
$a_{k-2}=-|E_{\Gamma}|$ of $\Lambda=\R^{k-1}$.
\end{thm}

\subsection{Wire networks}
We briefly recall the relevant notions from \cite{kkwk} which we will need in the examples.
As in the introduction, we fix a graph $\G$ embedded in $\R^n$, and a maximal translational symmetry group $L$.
such that $\bar\Gamma:=\Gamma/L$. Let $\pi:\Gamma\to\BG$ denote the projection.
We will directly set the magnetic field form $\Theta=0$.
Let $V_{\Gamma}$ be the set of vertices of $\Gamma$ and
$V_{\BG}$ be the vertices of $\BG$ and consider $\H=\ell^2(V_{\Gamma})=\bigoplus_{v\in \BG}\H_v$, where $\H_v=\ell^2(\pi^{-1}(v))$.
The translation group $L$ then naturally acts by translation operators on $\H$ preserving the summands.
This representation is then by commuting linearly independent unitaries and hence gives rise to a copy of $\Tn$.

Fix a spanning tree, with root $v_0$, and an order of the vertices.  Let $T_{\vec{e}}$ be the translation operator along a vector $\vec{e}$.
Notice that the oriented edges $\vec{e}$ lift to unique vectors in $\R^n$ under $\pi^{-1}$. Let $T_{v_iv_0}=T_{v_0v_i}^{-1}$
be the total translation along the unique shortest edge path in the spanning tree from $v_0$ to $v_i$.
Then $T_{v_iv_0}$ given an isometry $\H_0\to \H_{v_i}$. And hence we get an isometry 
$\H\simeq \bigoplus_{v\in \BG}\H_{v_0}:=\H_0$.

Then the weight of an oriented edge $\vec{e}$ from the vertex $v_i$ to the vertex $v_j$ is the translation operator
$ T_{v_0v_i}T_{\vec e}T_{v_iv_0}$. This defines the Harper Hamiltonian as the corresponding
graph Hamiltonian which acts on $\H_0$. It is shown in \cite{kkwk} that indeed these translation operators lie in the translations
generated by $L$ and hence give unitaries in $\Tn$.
Pulling back the translation operators of $L$ to $\H_0$, they together with $H$ generate the commutative
 $C^*$ algebra $\B_0$.

The physical background  for this data as explained in detail in \cite{kkwk,kkwk2} is as follows. Given one of the triply--periodic CMC surfaces, P (primitive), D (diamond) or G (Gyroid), one can consider its thickened or ``fat'' version. Its boundary then consists of two non--intersecting
surfaces, whence the name Double Gyroid, for instance. These surfaces give interfaces which appear
in nature. In particular, the Double Gyroid could recently be synthesized on the nano--scale \cite{Hillhouse}.
The structure contains three components, the ``fat'' surface or wall and two channels.
Urade et al.\ \cite{Hillhouse} 
have also demonstrated a nanofabrication technique in which the channels are 
filled with a metal, while the silica wall can be either left in place or removed. This yields
two wire networks, one in each channel. The graph we consider and call Gyroid graph
 is the skeletal graph of one of these channels. The  P, D, and G examples are the unique such surfaces
where the skeletal graph is symmetric and self--dual. The graph Hamiltonian is then the Harper Hamiltonian for one channel of this wire network.
 The 2d--analogue of this structure is the honeycomb lattice underlying graphene.

\section{Calculations}
Since all calculations are for the base $T^n$  we will consider the function $P$ locally pulled back via  the exponential map 
$\exp:\R^n\to T^n$
$(a_1,\dots, a_k)\to (\exp(i a_1),\dots, \exp( i a_n))$.
In this notation given a point $(a_i)$ and its corresponding character $\chi$,
the translation operators $U_j$ corresponding to the generators of the algebra $\Tn$ 
get mapped to $\chi(U_j)=\exp(ia_j)$. Indeed under the Gel'fand--Naimark correspondence
the operator $U_j$ is the function $\exp(ia_j)$ for the coordinates of $T^n$ above. To simplify
the calculation, we will drop the $\chi$ and just write the function for the operator.

\subsection{The Gyroid}
\label{gyroidsec}

\subsubsection{The matrix and the function $P$}
As shown in  \cite{kkwk}, the relevant graph for the Gyroid wire network
 is the full square which is simply laced. We also fix a spanning tree shown in Figure \ref{spt_gyroid}
ordering the vertices as indicated, the root being $1$.

 \begin{figure}
\includegraphics[width=.3\textwidth]{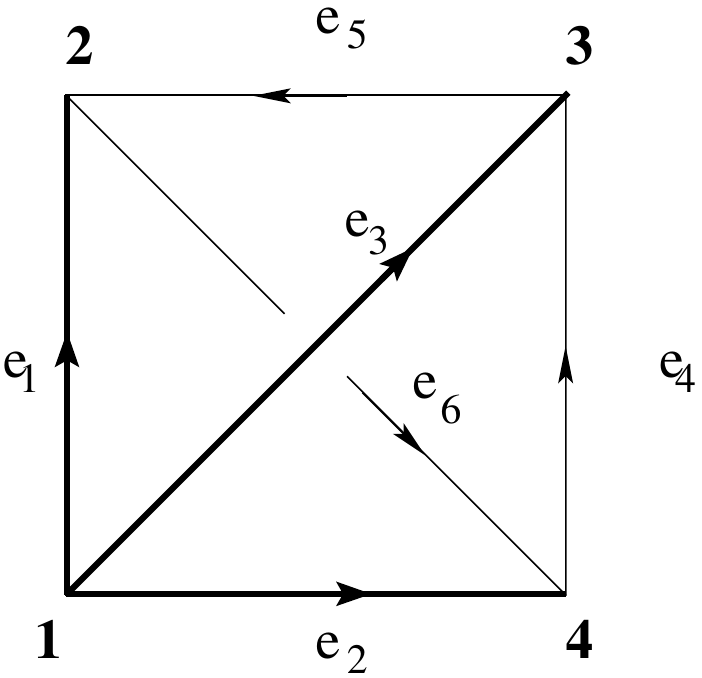}
\hspace{1cm}
\includegraphics[width=.5\textwidth]{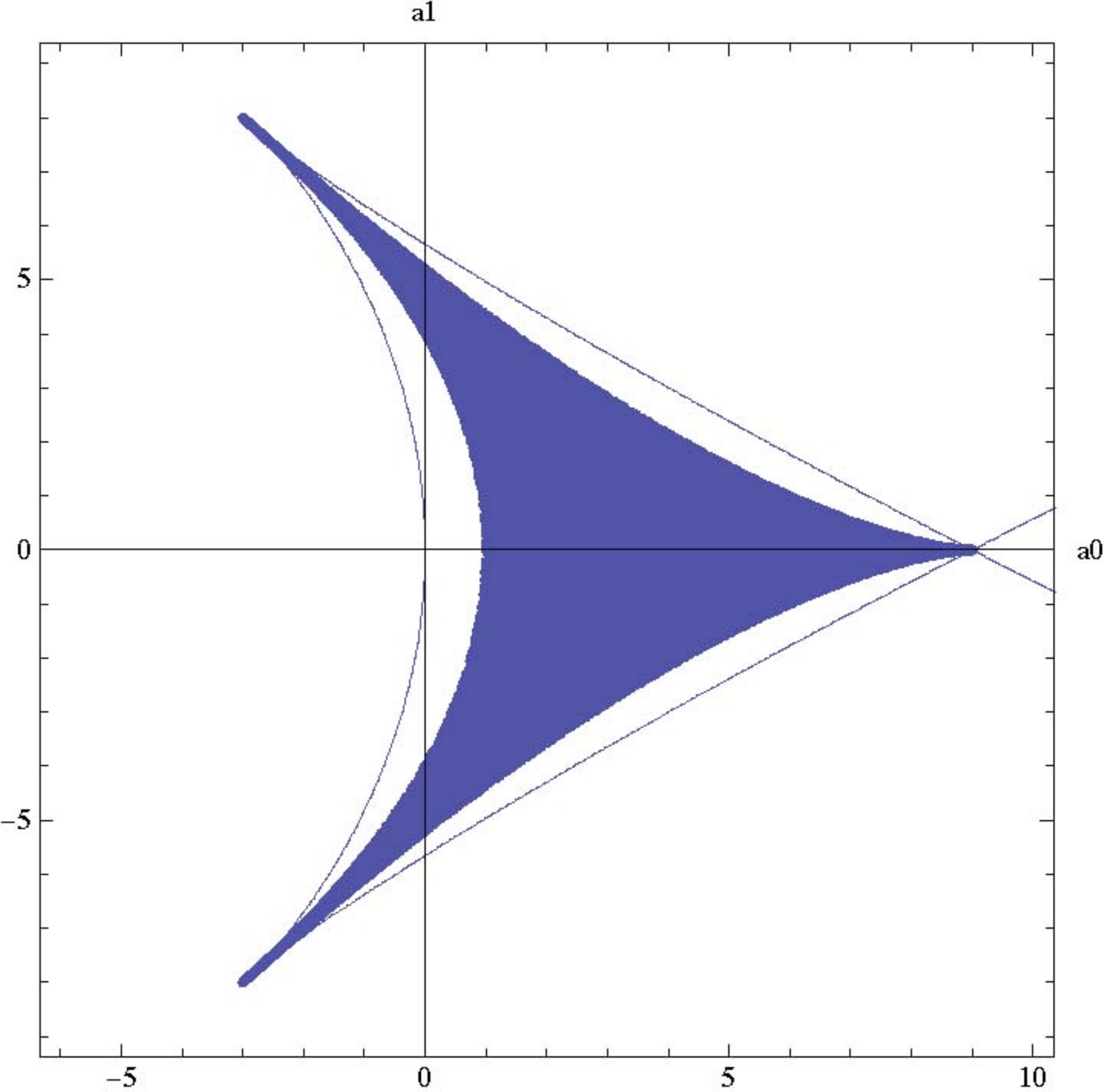}
\caption {Spanning tree and  characteristic region for the Gyroid (solid region). The curve
is the slice of the discriminant of the $A_3$ singularity at $a_2\equiv -6$}
 \label{spt_gyroid}
\end{figure}

The Harper Hamiltonian is a function on $T^3=S^1\times S^1\times S^1$. We can visualize $T^3$ as a cube where opposite sides are identified.

The Harper Hamiltonian reads \cite{kkwk}

\begin{equation}
H=\left(
\begin{array}{cccc}
0&1&1&1\\
1&0&A&B^*\\
1&A^*&0&C\\
1&B&C^*&0
\end{array}
\right)
\end{equation}
where $A$, $B$ and $C$ are operators generating $\T^3_0$ each of which we can think of as a function
on $S^1$ We will rewrite them as $A=\exp(i a)$, $B=\exp(i b), C=\exp(i c)$
with $a,b,c$ real.

The eigenvalues of $H$ are given by the roots of the characteristic polynomial:
\begin{equation}
P(a,b,c,z)=z^4-6z^2+a_1(a,b,c)z +a_0(a,b,c)
\label{charactpol}
\end{equation}
where
\begin{eqnarray}
\label{a0a1}
a_1&=&-2 \cos(a)-2 \cos(b)-2 \cos(c) -2 \cos(a+b+c)\nonumber\\
a_0 &=& 3-2\cos(a+b)-2 \cos(b+c) -2 \cos(a+c)\nonumber
\end{eqnarray}
give the characteristic map $\chr:=(a_0,a_1): T^3\to \R^2$

 \subsubsection{A quick look at the characteristic region.}
The characteristic region of the full square graph is depicted in  Figure \ref{spt_gyroid}.
The curve shown in the figure is the discriminant locus which is explicitly given by
\begin{equation}
 \label{disc}
 20736\; a_0 - 4608 \;a_0^2 + 256 \;a_0^3 + 864\; a_1^2 - 864\; a_0 \;a_1^2 - 27 \;a_1^4=0
 \end{equation} 
 
 The boundaries of the characteristic region are obtained as the collection of points $(a_0,a_1)$ for 
 $a=b=c$ and $a=b=-c$.

We see that the characteristic region is contained in the slice $a_2=-6$ of the $A_3$ singularity and intersects the discriminant
 in exactly three isolated points, the two cusps and the double point of that slice of the swallowtail.
 The  two cusps are in the stratum of type $A_2$ and the double point is in the stratum of type $(A_1,A_1)$.  As is quickly seen and we calculate below the fibers over all these points are indeed discrete. For the $A_2$ singularities, this is just one  point each, giving rise to two triple
 crossings, and the fiber over $(A_1,A_1)$ consists of two points. Over each of these points
 there are two double crossings and it turns out, see below, that these are
 Dirac points.

 This is a very special situation in that the points on the discriminant are
 actually at singular points of the region.
 From the singularities it is easily seen
 that the image of the tangent spaces at the points in the fiber is $0$ and hence $J_a$ vanishes.

\subsubsection{Classification of the critical points}
We first check the condition $\nabla P=0$. This yields the equations:
\begin{eqnarray*}
\frac{\del P}{\del a}&=&z (2 \sin (a+b+c)+2 \sin (a))+2 \sin (a+b)+2 \sin (a+c)=0\\
\frac{\del P}{\del b}&=&z (2 \sin (a+b+c)+2 \sin (b))+2 \sin (a+b)+2 \sin (b+c)=0\nn\\
\frac{\del P}{\del c}&=&z (2 \sin (a+b+c)+2 \sin (c))+2 \sin (a+c)+2 \sin (b+c)=0\nn\\
\frac{\del P}{\del z}&=&-2 \cos (a+b+c)-2 \cos (a)-2 \cos (b)-2 \cos (c)+4 z^3-12 z=0\nn
\end{eqnarray*}

To solve these equations, we rewrite the first three using trigonometric identities as
\begin{eqnarray*}
z \cos (\frac{b+c}{2})=-\cos (\frac{b-c}{2})\nn\\
z \cos (\frac{a+c}{2})=-\cos  (\frac{a-c}{2})\nn\\
z \cos (\frac{a+b}{2})=-\cos (\frac{a-b}{2})\nn
\end{eqnarray*}

In the case that all cosines are different from zero, we can solve each  equation for $z$ and set them equal. This leads to 
\begin{eqnarray}
 \cos (\frac{b-a+2c}{2}) = \cos (\frac{a-b+2c}{2})\nn\\
 \cos (\frac{2a-b+c}{2}) = \cos (\frac{2a+b-c}{2})\nn
 \end{eqnarray}
 
 From the first of these equations we get $a=b\; \mbox{mod 2} \pi$, from the second
 $b=c \;\mbox{mod 2}  \pi$, and $z \cos(a) =-1$ for $\cos(a) \neq 0$.  Plugging this back into the last of the original equations, we find
 $$8 \cos^6(a)+4-12 \cos^2(a)=0$$ Among the solutions we pick those for which the characteristic polynomial $P(a,b,c,z) $ (see Eq. (\ref{charactpol})) is zero,
 namely $\cos(a)=-1/z=\pm1$, i.e. $a=0,\pi$ and $z=\pm1$.
 
 In the case that $\cos(a)=0$ we obtain 
 $$ 4 z^3-12z=0$$ which has solutions $z=0, \pm \sqrt{3}$. $z=0$ has to be discarded since it does not satisfy $P(a,b,c,z) =0$.

Summing up,  the critical points are

\begin{enumerate}
\item $ a=b=c=0  \;\mbox(mod\;2 \pi); z=-1$
\item $a=b=c=\pi \;\mbox(mod\;2 \pi); z=1$
\item $a=b=c=\frac{\pi}{2},\frac{3 \pi}{2}\;\mbox(mod\;2 \pi); z=\pm \sqrt{3}$
\end{enumerate}

Looking at the image of these points under the characteristic map,
we see that $\chr(0,0,0)=(-9,-3)$ is the lower cusp, which is an $A_2$ point with a triple degeneracy,
so is $a(\pi,\pi,\pi)=(9,3)$. For the other two points,  $a(\pi/2,\pi/2,\pi/2)=a(3\pi/2,3\pi/2,3\pi/2)=(0,9)$
and they are in the $(A_1,A_1)$ stratum. This means that these points are candidates for Dirac points,
which they indeed are.

To decide this, we calculate
the Hessian.

Plugging in  $a=b=c=\frac{\pi}{2},\frac{3 \pi}{2}\;\mbox{mod}\;2 \pi;$ $ z=\pm \sqrt{3}$, it becomes
\begin{equation}
\mbox{Hess}=\left(
\begin{array}{cccc}
-4&-2&-2&0\\
-2&-4&-2&0\\
-2&-2&-4&0\\
0&0&0&24
\end{array}
\right)
\end{equation}
which has signature$(---+)$. Notice that the corresponding cone is also not tilted.
For the other two points, the Hessian vanishes, as expected.

\begin{figure}
\includegraphics[width=\textwidth]{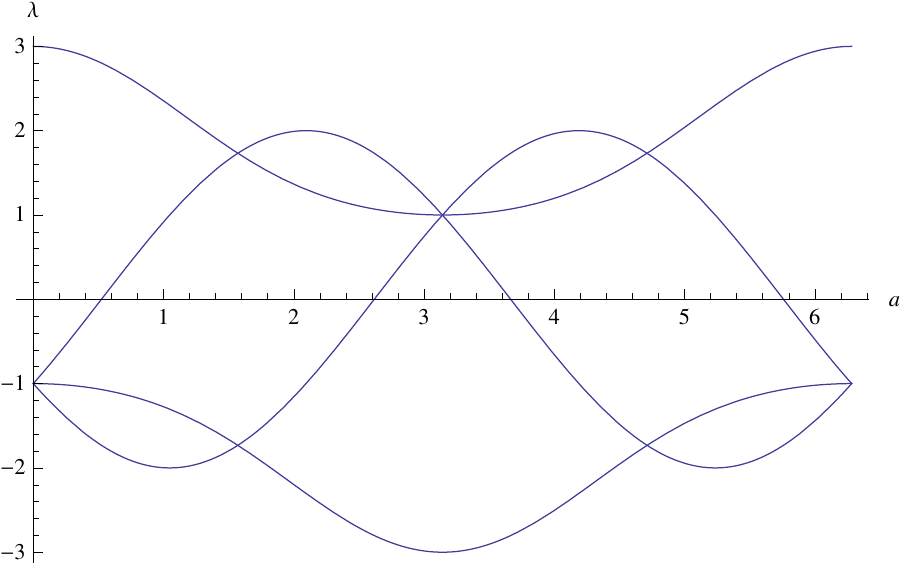}
\caption{Spectrum of the Gyroid Harper Hamiltonian for $a=b=c$}
\label{dispersiondiagonal}
\end{figure}

Since the singular points are all isolated, we can take any direction as transversal.
It turns out that the diagonal curve $C$  given by $a=b=c$ gives a transversal direction 
for all 4 singular points at once.
The spectrum on this line is given in Figure \ref{dispersiondiagonal}
can be obtained by using group theory cf.\ \cite{sym}. 
This is how it was first found in \cite{Avron} where the authors
considered an equivalent family in another context. They also found no other singularities numerically.
Explicitly the spectrum along this diagonal is given by
\begin{eqnarray}
\lambda_1=\omega \, exp(ia)+\bar{\omega} \exp(-ia)&& \lambda_2=\bar{\omega} \exp(ia)+\omega \exp(ia)\nn\\
\lambda_{3,4}=\cos(a)\pm \sqrt{\cos^2(a)+3}
\end{eqnarray}
From this one can see a linear dispersion relation in the direction of the diagonal. Without further
analysis, one cannot deduce the dispersion relation in any other direction from this results.

By our previous analysis we have, however, proven {\em analytically}, that there are indeed no other singularities and furthermore
have determined that there is a linear dispersion relation in {\em all directions} at the points $(\pi/2,\pi/2,\pi/2)$
and $(3\pi/2,3\pi/2,3\pi/2)$ establishing that these are indeed
two Dirac points.

It is interesting to note that the image curve $\chr(C)$ runs first from $(-9,-3)$ to $(0,9)$ along
the boundary of the region,
 continues as the boundary
curve to $(9,-3)$, and then turns back on itself to cover these two boundary pieces twice.

The  nature of the two triple points is that they are isolated and they are
a pullback of the unfolding of the $A_2$ singularity. This is given by the image
of the characteristic map where now we consider the local neighborhood of the cusp point
in the slice as a miniversal unfolding of $A_2$. This is indeed possible and a standard 
way of embedding the unfolding of $A_2$ into that of $A_3$ \cite{arnoldbook}.

\subsection{The honeycomb and diamond cases}
We treat the diamond and the honeycomb case in parallel.
The graphs $\BG$ are given in Figure \ref{graphsDhoney}.

\begin{figure}
\includegraphics[width=.5\textwidth]{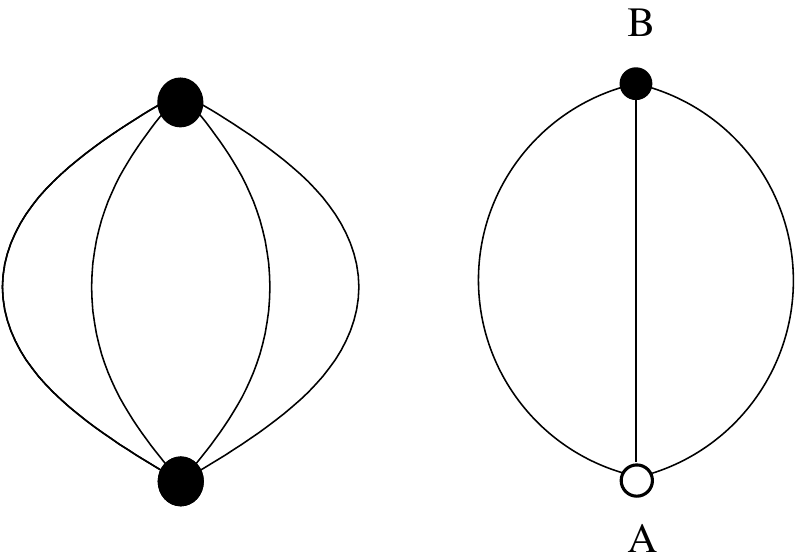}
\caption{The graphs $\BG$ for the diamond (left) and the honeycomb case (right)}
 \label{graphsDhoney}
\end{figure}

The Hamiltonians are 
\begin{equation}
H_{hon}=\left(
\begin{matrix}
0&1+U+V\\
1+U^*+V^*
\end{matrix}
\right)
\end{equation}
and
\begin{equation}
H_D=\left(
\begin{matrix}
0&1+U+V+W\\
1+U^*+V^*+W^*
\end{matrix}
\right)
\end{equation}
We again use $U=exp( i u),V=exp( i v),W=exp(i w)$.

The polynomials are $P(u,v,z)=z^2-3-2cos(u)-2cos(v)-2\cos(u-v)$ and
 $P(u,v,w,z)=z^2-4-2\cos(u)-2\cos(v)-2\cos(w)-2\cos(u-v)-2\cos(u-w)-2\cos(v-w)$.
The characteristic regions in $\R$ are just the intervals $[-9,0]$ and $[-16,0]$. The discriminant
is the point $0$. From this we see that in both cases we have to have $a_0=0$ and 
the singular locus is simply this fiber.

\subsubsection{The honeycomb case}
 For the honeycomb,
the standard calculation shows that in this case $U=V^*$ and $U\in \{\rho_3:\exp(2\pi i/3),\bar \rho_3\}$, which means that the fiber consists of 2 points.
These are the well known Dirac points $(\rho_3,\bar\rho_3),(\bar\rho_3,\rho_3)$.
We can check this explicitly:
$\nabla(P)=(2\sin(u)+2\sin(u-v),2\sin(v)-2\sin(u-v),z)$ from which we see that $z=0$ and $u\equiv - v 
\equiv 2u (2\pi)$. Furthermore 
\begin{equation}
Hess_{hon}=\left(
\begin{matrix}
2\cos(u)+2cos(u-v)&-2\cos(u-v)&0\\
-2(cos(u-v))&2\cos(v)-2\cos(u-v)&0\\
0&0&1\\
\end{matrix}
\right)
\end{equation}
which has the correct signature $(--+)$ at the given points, from which we recover
the known result that these points are Dirac points.

\subsubsection{The diamond case}
The equation for the fiber over $0$, $$-4-2cos(u)-2\cos(v)-2\cos(w)-2\cos(u-v)-2\cos(u-w)-2\cos(v-w)=0$$
has been solved in \cite{kkwk2} and the solutions are given by $(u,v,w) =(\phi_i,\phi_j,\phi_k)$
with $\phi_i=\pi, \phi_j\equiv \phi_k+\pi\; \mbox{mod}\; 2\pi$ with $\{i,j,k\}=\{1,2,3\}$.
So in this case the fiber of the characteristic map is 1--dimensional and the pull--back
has singularities along a locus of dimension $1$, which also implies that there are no Dirac points.
Geometrically the singular locus are three circles pairwise intersecting in a point.

\subsubsection{The characteristic  map and region}
For both the honeycomb and the diamond graph, the relevant singularity is $A_1$. Both these graphs
are not simply laced, so their image is not contained in a slice. The swallowtail is only one point $0$
and this is the stratum of type $A_1$.
In the honeycomb case the fiber over this point is discrete and consists of two points, while in the case of the diamond lattice the fiber is not discrete and it is given by three circles pairwise intersecting at a point. It turns out that in the honeycomb case the two candidates for Dirac points are indeed Dirac points. While for the $D$ case there is a non--trivial fiber which is essential 1--dimensional.
Hence we do not get Dirac points, but rather spread out singularities.

\subsection{Three--vertex graphs}

In order to have some examples for the $A_2$ singularity and to show
the kind of behavior that is possible, we considered
a three vertex graph with either only simple edges, or one, two or all of the edges doubled.

The characteristic regions are seen in Figures \ref{trigA}-\ref{scatterABCD}.
Here we see that for the simply laced case, we get a slice, for one or two doubled edges,
we get parabolic regions, which intersect the boundary in two points, which are of type $A_1$
and finally, in the case of all edges being double, the characteristic region is a surface,
which is bounded by the discriminant and the line on which $a_1$ takes its maximal value, in this case $a_1=12$.

\subsubsection{Triangle with single bonds}
We consider the graph and the spanning tree given in Figure \ref{trigA}.
 \begin{figure}
\includegraphics[width=.3\textwidth]{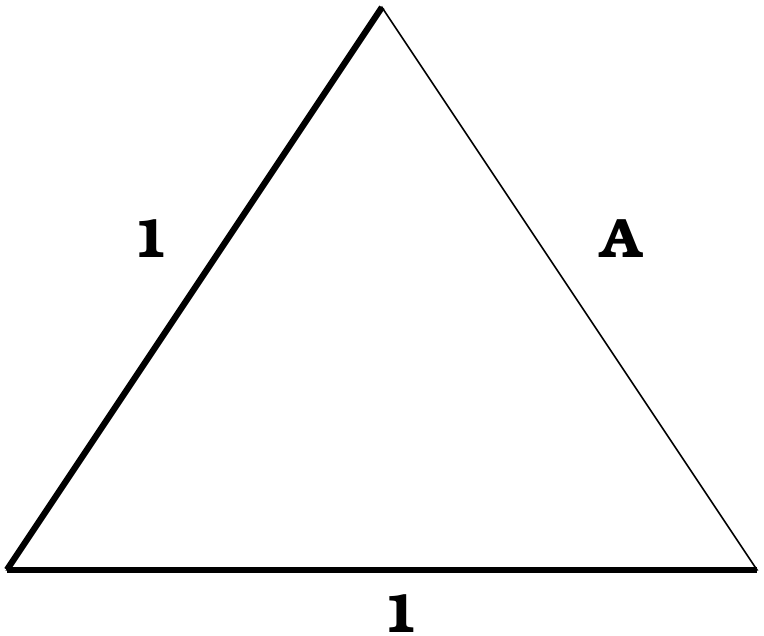}
\hspace{1cm}
\includegraphics[width=.3\textwidth]{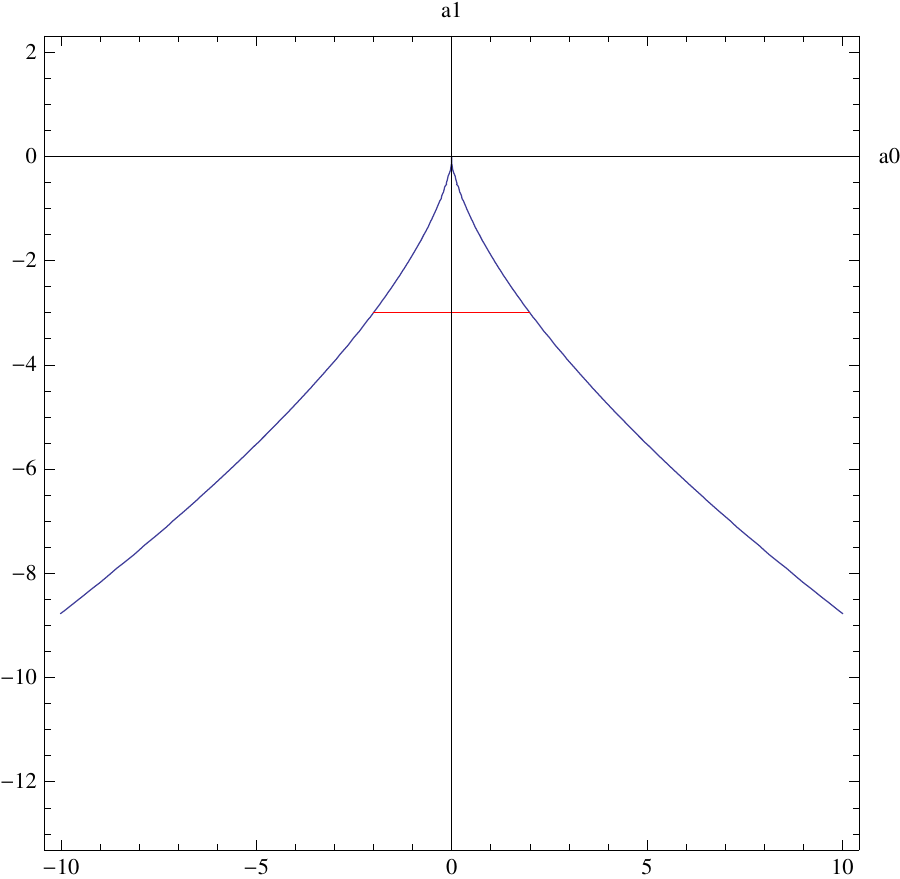}
\caption {Spanning tree and characteristic line for a triangle with single bonds}
 \label{trigA}
\end{figure}

The associated Harper Hamiltonian reads:
\begin{equation}
H=\left(
\begin{array}{cccc}
0&1&1\\
1&0&A\\
1&A^*&0
\end{array}
\right)
\end{equation}
where $A$ is an operators on $S^1$. We will rewrite it as $A=\exp(i a)$ with $a$ real.
The characteristic polynomial is:
\begin{equation}
P(a,z)=z^3-3z-2 \cos(a)
\end{equation}

The characteristic region is easy to calculate, since the graph is simply laced
it is contained in the slice $a_1=-3$. The image under $-2\cos(a)=[-2,2]$ thus  $R=[-2,2]\times -3$ in $\R^2$.
Figure \ref{trigA} shows this 
together with the zero locus of the discriminant.

From this we see that all possible singularities occur at $a_0=-2\cos(a)=\pm2$ that is $a\equiv 0,\pi (2\pi)$. Indeed  calculating $\nabla  P(a,z)$, we use
\begin{eqnarray}
\frac{\del P}{\del a}&=&2 \sin(a)\\
\frac{\del P}{\del z}&=&3 z^2-3
\end{eqnarray}

These equations vanish simultaneously for the following choice of variables:
\begin{enumerate}
\item $z=\pm1$
\item $a=0, \pi \; (\mbox{mod} \;2 \pi)$
\end{enumerate}

Among all possible combinations, the choices $(z=-1,a=0)$ and $(z=1,a=\pi)$ are zeros
of $p(z,a)$.  For these two points, the Hessian
\begin{equation}
\mbox{Hess}=\left(
\begin{array}{cccc}
2 \cos(a)&0\\
0&6z
\end{array}
\right)
\end{equation}
has a non-vanishing determinant of $\mbox{det(Hess)}=-12$ and signature $(-+)$.
So we find two Dirac points. Again the cone is not tilted.

\subsubsection{Triangle with one double bond}
We consider the graph and the spanning tree given in Figure \ref{scatterAB} where one of the bonds is a double bond.
The Harper Hamiltonian reads in this case
\begin{equation}
H=\left(
\begin{array}{cccc}
0&1&1\\
1&0&A+B\\
1&A^*+B^*&0
\end{array}
\right)
\end{equation}
where $A$, $B$ are operators on $S^1$. We will rewrite them as $A=\exp(i a)$, $B=\exp(i b)$ with $a,b$ real.
The characteristic polynomial is:
\begin{equation}
P(a,b,z)=z^3-(4+2\cos(a-b))z-2 \cos(a)-2 \cos(b)
\end{equation}
Again, we set $a_1=-(4+2\cos(a-b))$ and $a_0=-2 \cos(a)-2 \cos(b)$.
We see that this time the characteristic region is not contained in a slice, which 
was not to be expected since the graph is not simply laced.
The region is depicted via a scatter plot in Figure \ref{scatterAB}. One reads off
that $R$ intersects with the discriminant locus in two points.
 \begin{figure}
\includegraphics[width=.3\textwidth]{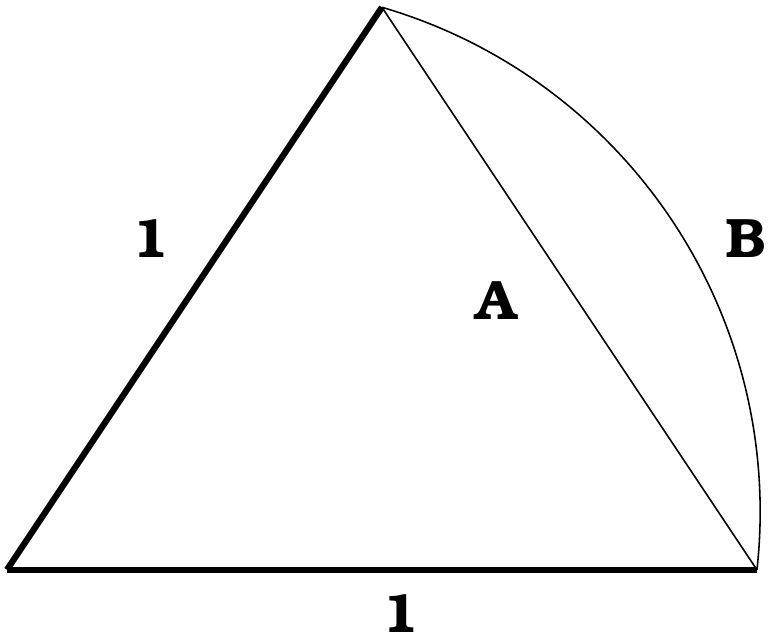}
\hspace{1cm}
 \includegraphics[width=.5\textwidth]{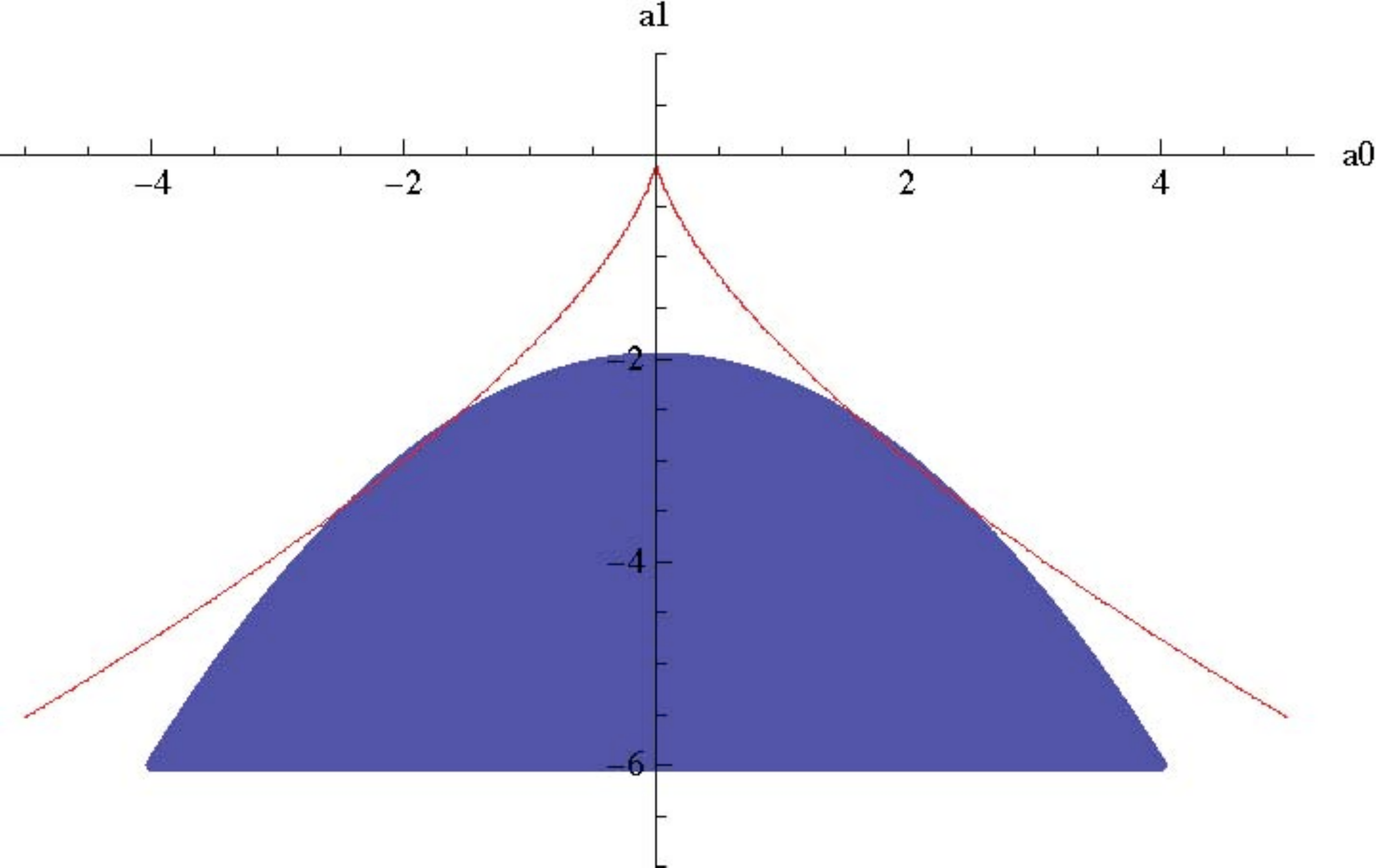}
\caption{Spanning tree and characteristic region with double bond}
 \label{scatterAB}
\end{figure}

To calculate   $\nabla  P(a,b,z)=0$ we use,
\begin{eqnarray}
\frac{\del P}{\del a}&=&2 \sin(a-b)z+2 \sin(a)\nn\\
\frac{\del P}{\del b}&=&-2 \sin(a-b)z+2 \sin(b)\nn\\
\frac{\del P}{\del z}&=&3 z^2-(4+2\cos(a-b))
\end{eqnarray}
From the first two equations, we get either $a=0,\pi$ and $b=0,\pi$, but for all combinations of those, the remaining two equations ($\frac{\del P}{\del z}=0$ and $P(a,b,z)=0$) cannot be simultaneously solved. Therefore the only possible solution for the first two equations is to take $a=-b$ and $z=-\frac{\sin(a)}{\sin(2a)}$ for $\sin(2a)\neq 0$. Putting this  into the last equation yields the trigonometric equation 
$$
3 \sin^2(a) -4 \sin^2(4a)-2\cos(2a)\sin^2(2a)=0
$$ 
which has the solutions $a= \pm \frac{ \pi}{3},\pm \frac{2 \pi}{3}$. These also  lead to a vanishing of the characteristic polynomial.

So we get the following two solutions (the negative values lead to the same values for $a_0$ and $a_1$):
\begin{enumerate}
\item $a=\frac{ \pi}{3}, b=-\frac{ \pi}{3} \;\mbox{mod}\; 2 \pi, z=-1$
\item $a=\frac{ 2 \pi}{3}, b=-\frac{ 2 \pi}{3}\;\mbox{mod} \;2 \pi, z= 1$
\end{enumerate}

The Hessian is
\begin{equation}
\mbox{Hess}=\left(
\begin{array}{cccc}

2 \cos(a-b)z+2 \cos(a)&-2z\cos(a-b)&2 \sin(a-b)\\
-2z\cos(a-b)& 2 \cos(a-b)z+2 \cos(b)&-2 \sin(a-b)\\
2\sin(a-b)&-2\sin(a-b)&6z\\
\end{array}
\right)
\end{equation}
and has signatures $(++-)$ and $(--+)$, respectively. Here the cone  is actually tilted.
Again there are two Dirac points.

\subsubsection{More variations}

In the same way, we can obtain information about possible Dirac points for the following graphs:

\begin{figure}[ht]
\centering
\includegraphics[width=.3\textwidth]{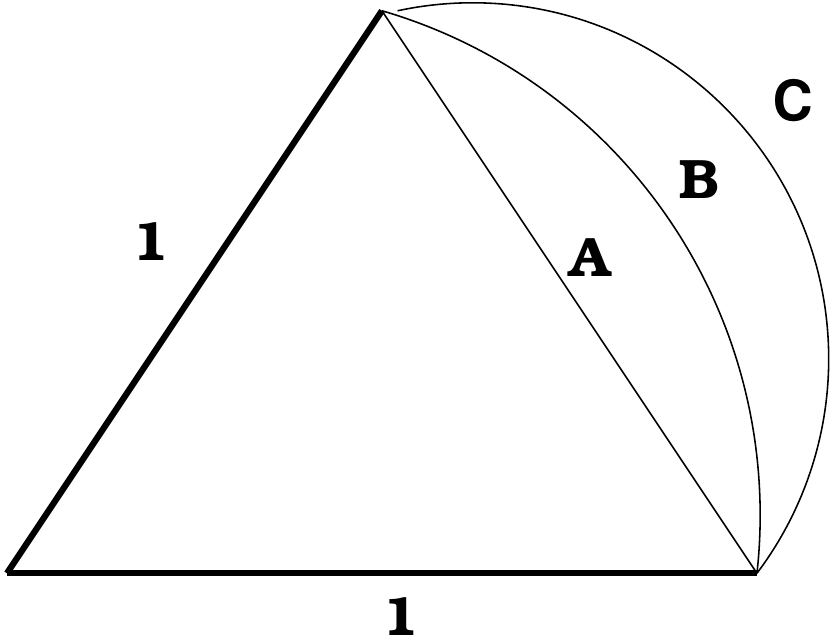}
\hspace{1cm}
\includegraphics[width=.5\textwidth]{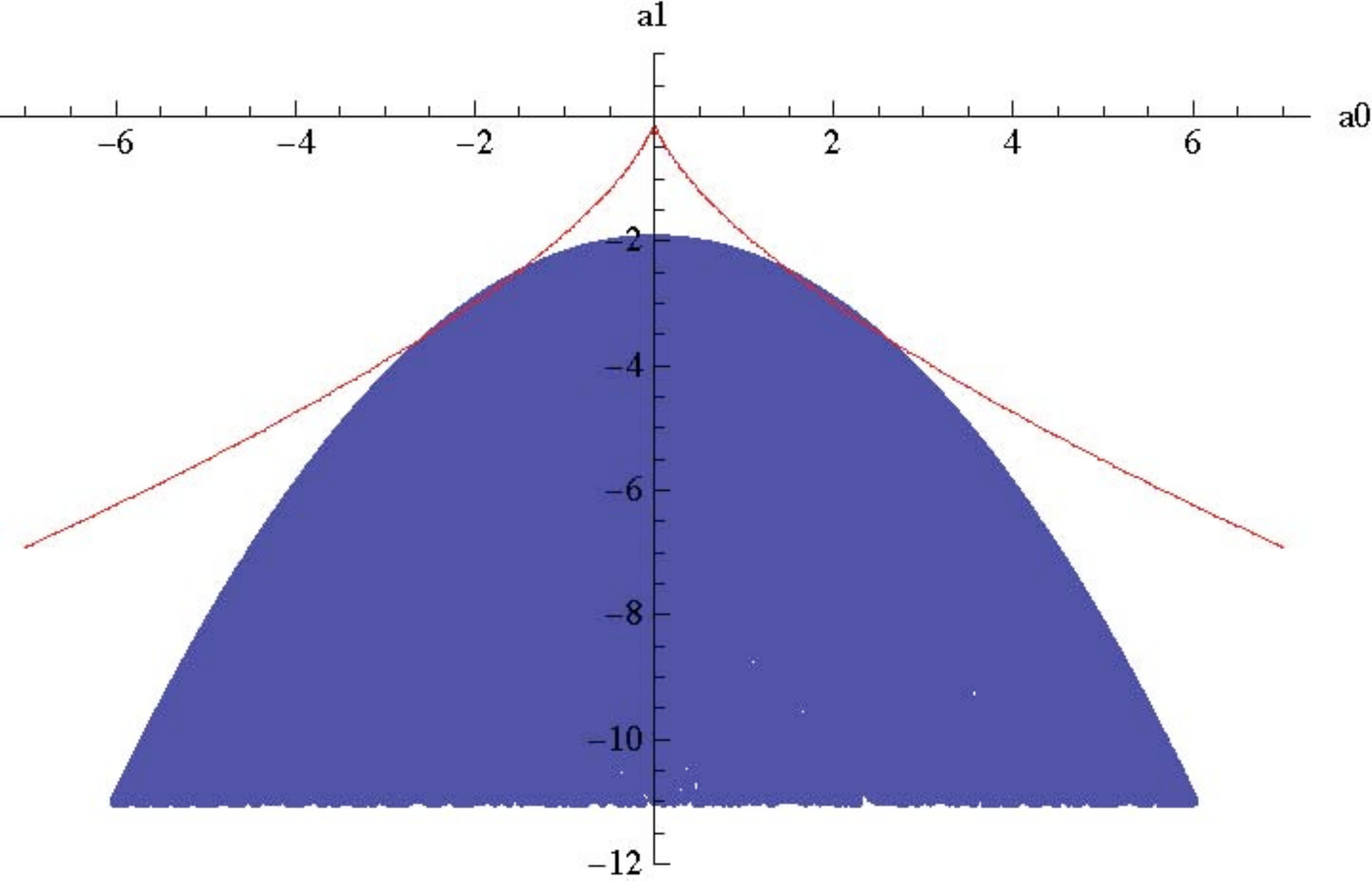}
\caption{Spanning tree and scatter plot of characteristic region with triple bond}
 \label{scatterABC}
\end{figure}

\begin{figure}[ht]
\includegraphics[width=.3\textwidth]{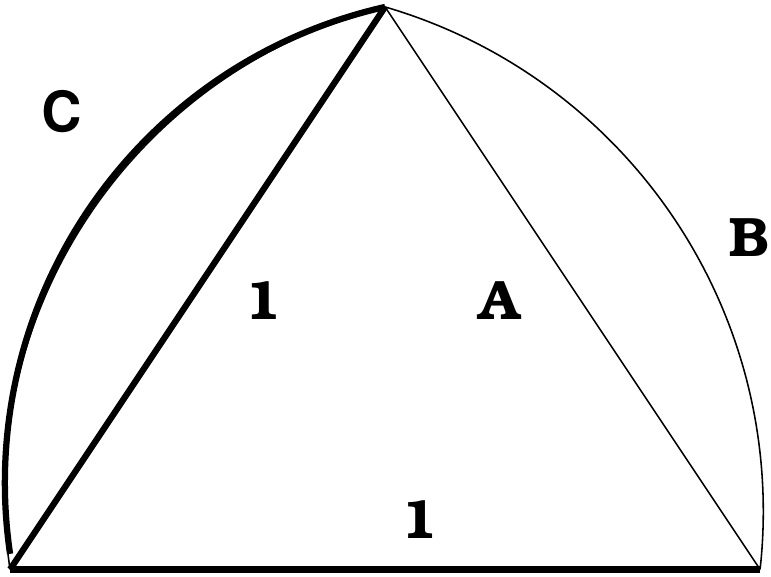}
\hspace{1cm}
\includegraphics[width=.5\textwidth]{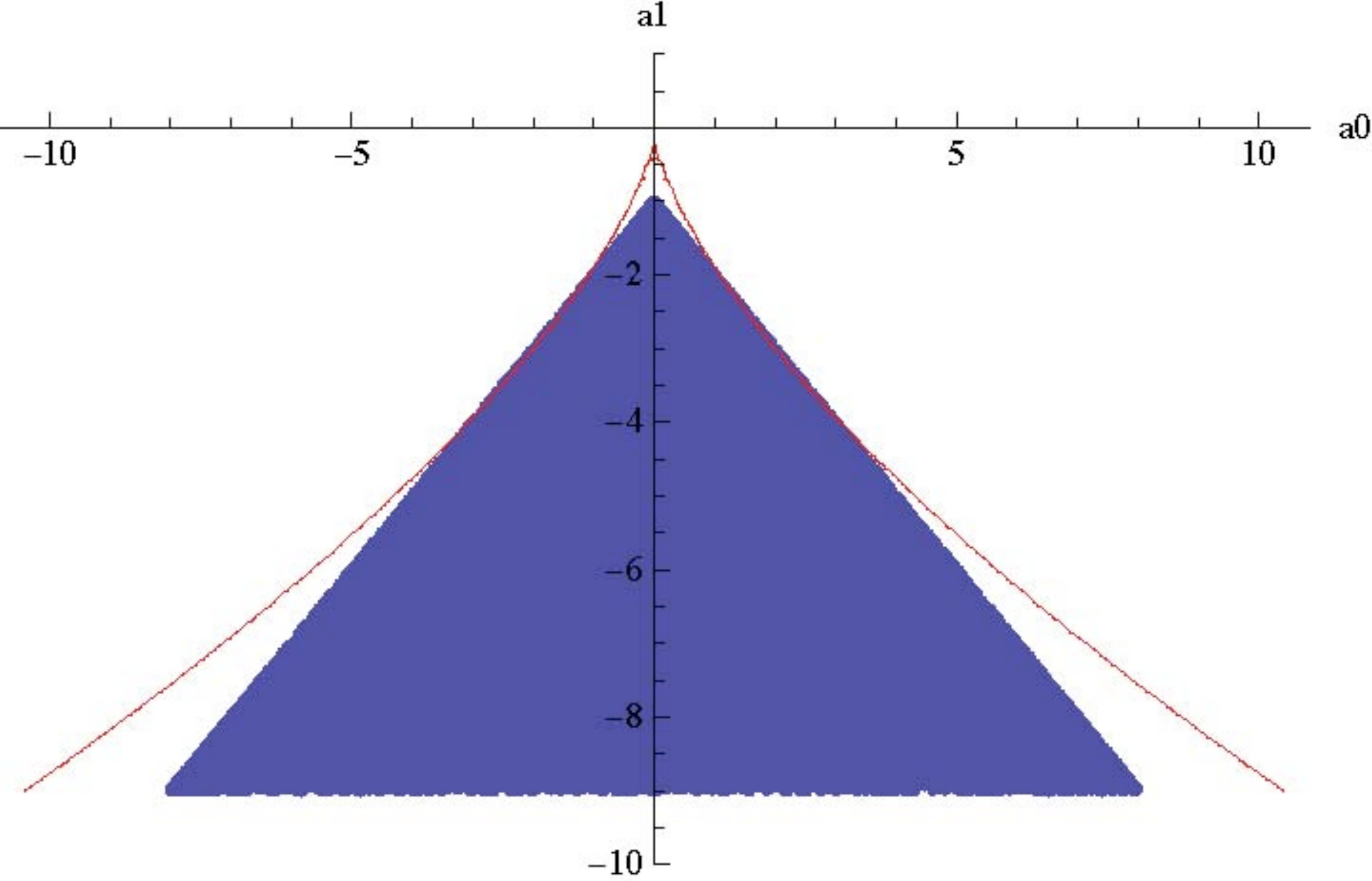}
\caption{Spanning tree and characteristic region with two double  bonds}
 \label{scatterABD}
\end{figure}

\begin{figure}[ht]
\includegraphics[width=.3\textwidth]{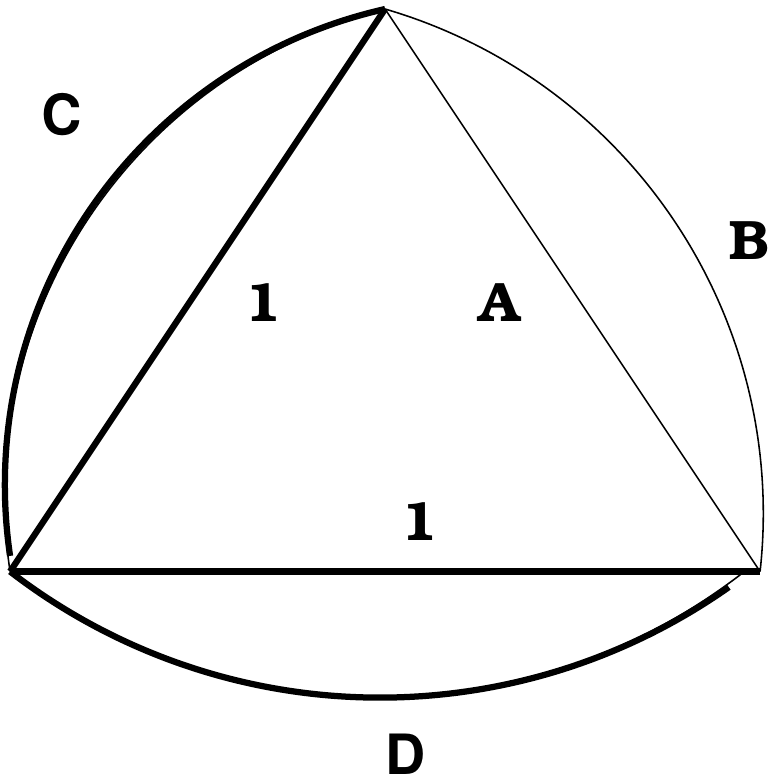}
\hspace{1cm}
\includegraphics[width=.5\textwidth]{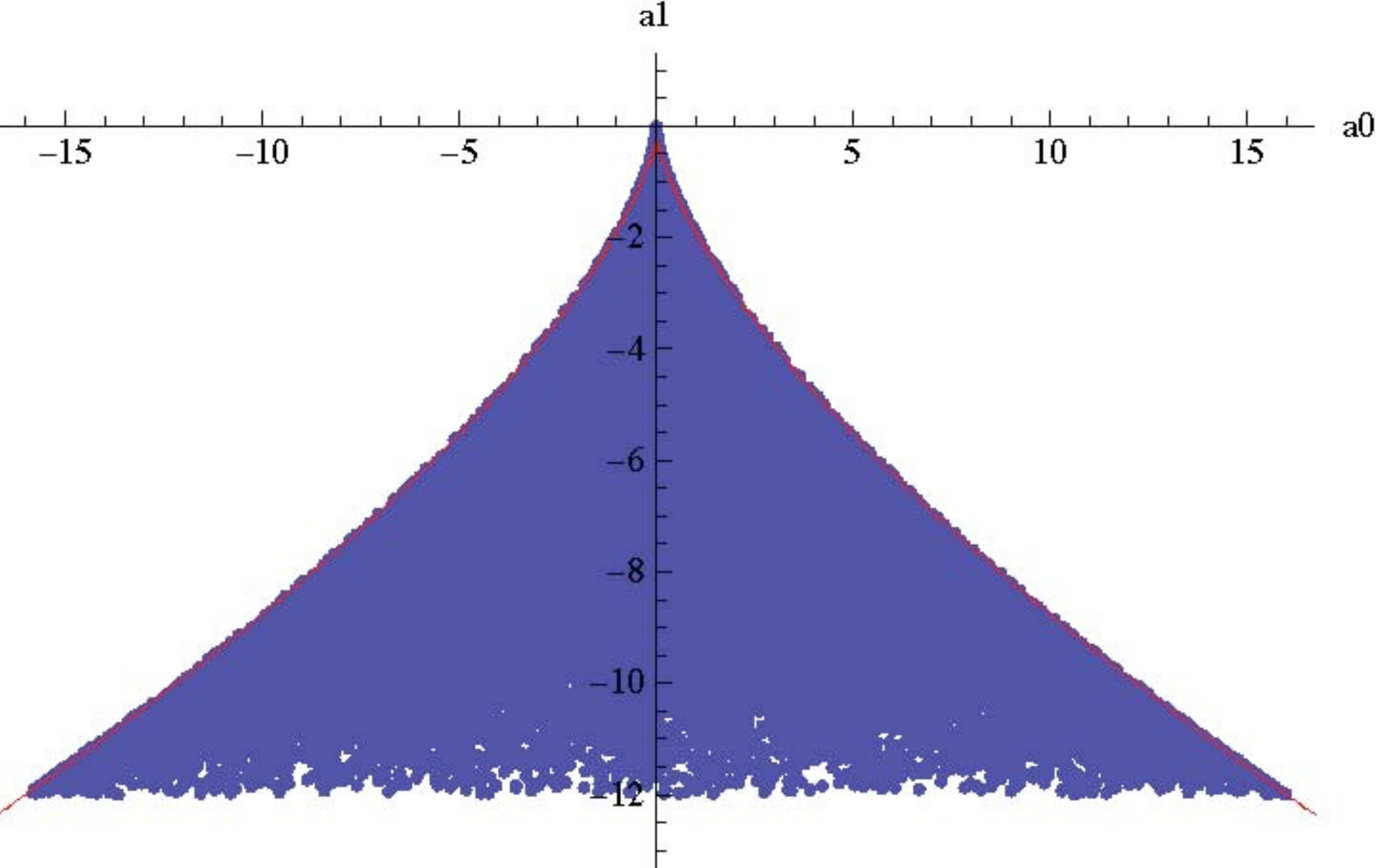}
\caption{Spanning tree and  scatter plot of characteristic region with three double  bonds}
 \label{scatterABCD}
\end{figure}

For the triple bond case shown in Figure \ref{scatterABC} and the two double bonds in Figure \ref{scatterABD}, 
we see two isolated intersection points, but the fiber will be dimension
$1$, so like in the D--case there will be no Dirac points.

Considering three double bonds (see Figure \ref{scatterABCD})
we see that now the intersection of $R$ with $D$ is along two of the boundaries
of $R$ and hence not only will the  fiber have dimension $2$, but we also expect
to have horizontal directions which to not resolve the singularity. We leave this for further
investigation.

\subsection{The P and other Bravais cases}
Here the graph has small loops, and the characteristic polynomial has to be transformed.
It is simply a polynomial of degree $1$. $P(t,z)=z-\sum_i t_i$, so that after shifting $z$ we
are left with just $z=0$, which is not critical. Not surprisingly, there are no singularities.

\section{Conclusion and outlook}
We have developed a general method to analyze singularities in the spectra of smooth families of $k\times k$
Hamiltonians parameterized
by a base $B$
using singularity theory. In particular, we realized the spectrum as a singular submanifold $X$ of the smooth product $B\times \R$
as  the zero set of a function $P$. 
This led us to a simple characterization of Dirac points as $A_1$ or Morse singularities of $P$ with critical value $0$ whose signature is
$(-\dots-+)$ or $(+\dots +-)$ . Furthermore we could represent $\pi:X\to B$ up to a given diffeomorphism on the ambient space
as the pull--back of the miniversal unfolding of the $A_{k-1}$ singularity via a characteristic map $\chr$. This classifies all
the possible singularities of the fibers of $\pi$ as $(A_{n_1},\dots, A_{n_l})$ 
The image of the characteristic map, called the characteristic region, allows one to read off which ones occur.

We then applied these techniques to the graph Hamiltonians and wire networks. Here we reproduce the known Dirac points
for graphene and make the surprising find that the Gyroid wire network also has  Dirac points. We expect that
this should have practical applications.

The situation for the Gyroid is very special, as the characteristic region goes into the cusps and the self--intersection locus of the swallowtail 
without prior contact to the ``walls''.  Were this not be the case, one would not expect isolated points.
We gave more graph examples to illustrate how special this behavior is.
Adding multiple edges, we expect to get ``full'' region as in case of the three double bonds in a triangle.

One exciting find is that there seems to be a 
commutative/non--commutative
duality in the wire network families first stated in \cite{kkwk2}. 
By this we mean the observation that there is a correspondence between the
locus of degenerate points in the commutative setting and
the locus of parameters where the corresponding non--commutative algebra 
$\B_{\Theta}$ is not the full matrix algebra $M_k(\TTheta)$.
The correspondence is not 1--1, but the top dimensions agree and
there are further features that look dual. We wish to emphasize that 
although the space $T^n$ appears in both settings as the parameter
space, it {\it a priori} plays two different roles. 
In the case without
magnetic field, the parameters are quasimomenta, while in the case with magnetic field 
they are the field's components.
It is intriguing to speculate that the non--commutative setting is
a model for a non--commutative unfolding of singularities and
that this furnishes the framework to make the duality explicit. 
This will be a topic of further research.

There are several other directions of research that are immediate. First one can ask if in the graph case there are symmetries forcing
the degeneracies. This is pursued in \cite{sym} using a regauging groupoid action.
Second, a physically relevant question is how stable the singular fibers are with respect
to deformations of the Hamiltonian. For a 3-dimensional parameter space, simple singularities, i.e., 
Dirac points, are ``magnetic monopoles'' \cite{Berry} and are expected to be topologically stable. This,
and the evolution of other types of singularities, is further discussed in \cite{impure}.  
We will also focus on making the
type of analysis explicit in the non--commutative geometry language. 
There should be some kind of characteristic classes and parings much like in the
setting of the quantum Hall effect as presented in \cite{BE,Marcolli}.

\section*{Acknowledgments}
RK thankfully acknowledges
support from NSF DMS-0805881.
BK  thankfully acknowledges support from the  NSF under the grant PHY-0969689.

  Any opinions, findings and conclusions or
recommendations expressed in this
 material are those of the authors and do not necessarily
reflect the views of the National Science Foundation.

Part of this work was completed when RK was visiting the IAS in Princeton, 
the IHES in Bures--sur--Yvette, the Max--Planck--Institute in Bonn and the University of Hamburg with a Humboldt fellowship. He gratefully acknowledges
their contribution. Likewise BK extends her gratitude to the Physics Department 
of Princeton, where part of this work was completed and to the DESY theory group where the finishing touches for this article were made.

The authors furthermore thank D. Berenstein, A. Libgober, M.~Marcolli and T. Spencer 
for discussions which were key to formalizing and finalizing our concepts.

\end{document}